\newif\ifthese\thesefalse
\documentclass{lmcs}
\pdfoutput=1

\usepackage{lastpage}
\lmcsdoi{15}{3}{12}
\lmcsheading{}{\pageref{LastPage}}{}{}%
{Jul.~02,~2018}{Aug.~07,~2019}{}

\usepackage{epsfig}

\usepackage[T1]{fontenc}

\usepackage {amssymb,amsbsy,latexsym,amstext}
\usepackage{stmaryrd}
\newif\iflong\longtrue
\newif\ifnot\notfalse

\usepackage{allmacros-lmcs} 
\usepackage{amsthm}
\usepackage{xspace}
\usepackage{centernot}

\usepackage{pgf,tikz}
\usetikzlibrary{arrows}

\usepackage{mathtools}
\usepackage{mathpartir}
\usepackage{etex}

\usepackage[all]{xy}

\usepackage{thmtools}


\usepackage{hyperref}
\theoremstyle{plain} 

\begin{document}

\title[Divergence and unique solution of equations]{Divergence and unique solution of equations}
\titlecomment{{\lsuper*} This is an extended version of the paper with
  the same title, published in the proceedings of CONCUR 2017~\cite{DBLP:conf/concur/DurierHS17}.}

\author[A.~Durier]{Adrien Durier}
\address{Univ. Lyon, ENS de Lyon, CNRS, UCB Lyon 1, LIP}
\email{adrien.durier@ens-lyon.fr, daniel.hirschkoff@ens-lyon.fr}

\author[D.~Hirschkoff]{Daniel Hirschkoff}

\author[D.~Sangiorgi]{Davide Sangiorgi}
\address{INRIA and Universit{\`a} di Bologna}
\email{davide.sangiorgi@cs.unibo.it}








\newcommand{\Div}{\mathit{Div}}
\newcommand{\Divi}{\mathit{Div}_{\omega}}
\newcommand{\ind}{\mathsf{ind}}
\newcommand{\ndesc}{\mathsf{nd}}

\theoremstyle{plain}
\newtheorem{theorem}[thm]{Theorem}
\newtheorem{proposition}[thm]{{Proposition}}
\newtheorem{lemma}[thm]{Lemma}
\newtheorem*{theorem*}{Theorem}

\theoremstyle{definition}
\newtheorem{definition}[thm]{Definition}
\newtheorem{example}[thm]{Example}
\newtheorem{remark}[thm]{Remark}

\newcommand{\daniel}[1]{\marginpar{\textcolor{blue}{DH: #1}}}
\newcommand{\danielmain}[1]{{\textcolor{blue}{DH: #1}}}

\definecolor{dbrown}{RGB}{200,100,0}


 \begin{abstract}
We study  proof techniques for bisimilarity based on \emph{unique solution of
  equations}. 



We draw inspiration from a result 
by
 Roscoe in the denotational setting of CSP and for failure semantics, 
 essentially  stating that an equation (or a system of equations) whose
infinite unfolding  never produces  a divergence has the unique-solution property. We  transport this result onto the operational setting of CCS and for
bisimilarity. 
We then exploit the operational approach to: 
refine the theorem, distinguishing between different forms of divergence; derive an
 abstract formulation of the theorems, on generic LTSs; adapt the theorems to other
 equivalences such as trace equivalence, and to preorders such 
 as trace
 inclusion. We compare the resulting techniques to enhancements of the bisimulation proof
 method (the `up-to techniques'). 
Finally, we study the theorems in name-passing calculi such as the asynchronous
$\pi$-calculus, and  
use them  to
 revisit the completeness part of the proof of full
 abstraction
of Milner's encoding of the $\lambda$-calculus into the $\pi$-calculus for L{\'e}vy-Longo Trees. 
 \end{abstract}
 \maketitle

\section{Introduction}


In this paper we study the technique of \emph{unique solution of
  equations} for (weak) behavioural relations. We mainly focus on 
bisimilarity but we also consider 
 other equivalences, such as trace equivalence, as well as
 preorders such as trace inclusion. 
Roughly, the technique consists in proving 
that two tuples of processes are
componentwise in a given   behavioural relation
 by establishing that 
 they are solutions 
of the same system of equations.

In this work, 
behavioural relations, hence also bisimilarity, are meant to be 
  \emph{weak} because  they abstract from internal moves of terms, as
  opposed to 
the  \emph{strong} relations, which
 make no distinctions between the 
 internal moves and  the
external ones (i.e., the interactions with the 
environment). Weak equivalences are, practically, the most
relevant ones:  e.g., two equal programs may produce the same result
with different numbers of evaluation steps. Further, the problems tackled in
this paper only arise in the weak case.

The technique of unique solution has been proposed by Milner
in the setting of CCS, and plays a  prominent role
in proofs of examples in his 
 book \cite{Mil89}. 
The method  is important in verification techniques and tools
  based on 
algebraic      reasoning \cite{RosUnder10,BaeBOOK,groote}. 
Not all equations have a unique solution: for instance any process
trivially satisfies $X =X$.
To ensure unique solution, it is often required that equations
  satisfy some kind of \emph{guardedness} condition; see, e.g., \cite{acp,groote}.
In CCS, the notion of guardedness is  syntactic:  variables of the equations can
only appear underneath a visible prefix. This 
is not sufficient to guarantee uniqueness. Hence, 
in Milner's theorem~\cite{Mil89}, uniqueness of solutions relies on an additional limitation:
 the equations must be
`sequential', that is, the variables of the equations may
not be preceded, in the syntax
tree, by the parallel composition operator.
This limits the expressiveness  of the  technique
(since occurrences of other operators above the variables, such as 
parallel composition and restriction, in general 
cannot be removed),  and its 
transport onto other languages
 (e.g.,  languages for distributed systems or
higher-order languages  usually do not include the sum operator,
which makes the theorem essentially useless).
 A comparable technique, involving similar 
limitations, has been proposed
 by Hoare in his book about CSP~\cite{hoare}, and plays an equally 
 essential role in the theory of CSP.

In order 
 to overcome such limitations, a variant of the technique, called 
\emph{unique solution of contractions},  has been proposed~\cite{popl15}. The technique is for behavioural equivalences;
however the meaning of `solution' is defined in terms of the
contraction of the chosen equivalence. Contraction is, intuitively,  a 
preorder that  conveys an idea of efficiency on
processes, where efficiency is measured on the number of internal
actions needed to perform a certain activity.
The 
condition   for applicability  of the technique is, as for
Milner's, purely syntactic: 
each
 variable in the body of an equation should be
 underneath a prefix.  The technique has two main disadvantages: 
{
1. the equational theory of the contraction preorder associated to an 
equivalence is not the same as the equational theory of the equivalence itself, 
which thus needs to be studied as well; 
2. the contraction preorder is
strictly finer than the equivalence,
hence there are equivalent processes, one of which might be solution of a given 
contraction, while the other is not, 
and the technique might not be applicable in this case.
}

In this paper we explore a different approach, inspired by results by Roscoe in CSP \cite{roscoe2,roscoe1}, 
essentially  stating that a guarded equation (or system of equations) whose
infinite unfolding  never produces  a divergence has the
unique-solution property. Roscoe's result 
is presented,  as  usual in 
CSP, with respect to  denotational semantics and {failure based
equivalence~\cite{failure,failure-divergence}.
%
In such a setting,
where 
divergence is catastrophic (e.g., it is the bottom element of the domain),
the  theorem has  an elegant and natural formulation. 
(Indeed, Roscoe develops a   denotational model  \cite{roscoe1} 
in which the  proof of the  theorem  is just a few lines.)

We draw inspiration from
Roscoe's work to formulate the counterpart of these results in the
operational setting of CCS and bisimilarity.
In comparison with the denotational CSP proof, 
the operational CCS proof  is more complex. 
The operational setting offers however a few advantages.
First, 
we can formulate  more refined versions of the theorem,  in which we
distinguish between different forms of divergence. 
\iflong
Notably, we can ignore 
divergences that appear after finite unfoldings of the equations, 
called \emph{innocuous} divergences in this paper.  
\fi
(These refinements would
 look less 
 natural in the  denotational and trace-based
setting of CSP, where  any divergence causes a process to
be considered undefined.)
A second and more important advantage 
comes as a consequence of 
the flexibility of the 
  operational approach: 
the unique-solution theorems 
 can be 
 tuned to other behavioural relations (both equivalences and
preorders), and to other languages. 

To highlight  the latter aspect,  we present   abstract  formulations of the
 theorems, on a generic LTS (i.e.,  without reference to CCS).
\iflong
 In this
 abstract formulation, 
\else
, where
\fi
the body of an equation becomes a function on the
 states of the LTS.  
The theorems for CCS are instances of the
 abstract formulations. 
%
%
Similarly we can derive analogous theorems for other
 languages. Indeed   we can do so 
  for all languages whose constructs  have an
operational semantics with rules in the 
 GSOS format \cite{gsos} (assuming appropriate hypotheses, among which
 congruence properties). In contrast, 
the analogous 
 theorems fail for languages whose constructs follow the
 \emph{tyft/tyxt} \cite{tyft} format, due to the possibility of rules with
a \emph{lookahead}.
We also consider extensions of the theorems to name-passing calculi such as
the $\pi$-calculus. 

The abstract version 
 of our main unique-solution theorem  has been formalised using
 the Coq proof assistant~\cite{www:usol:coq}.

 The bisimulation proof method is a well-established and extensively 
studied
technique for deriving bisimilarity
results. 
An advantage of this method are  its 
enhancements,
provided by the so called `up-to techniques'~\cite{pous:sangiorgi:upto}.
Powerful such enhancements are  `up to context', whereby  in the
derivatives of two terms a common context can be erased,  
 `up to expansion', whereby two
 derivatives can  be rewritten using the expansion preorder, and 
 `up to transitivity', whereby  the matching between two
derivatives is made with respect to the transitive closure of the candidate
 relation (rather than the relation alone). Different enhancements can sometimes be combined, though
care is needed  to preserve 
 soundness.  
One of the most powerful combinations is an \emph{`up to transitivity
and context'} 
technique due to Pous, which relies on a termination
hypothesis~\cite{damien_phd}. This technique generalises `up to expansion' and combines
it with `up to context' and 
`up to transitivity'.  
We show that, modulo
an additional hypothesis, our techniques are at least as
powerful as this up-to 
technique: any  up-to relation can
be turned into a system of equations of the same size 
\iflong
 (where the size of a relation is the
number of its pairs, and the size of a system of equations is the
number of its equations) for which uniqueness of solutions holds.
\fi

An important difference  between unique solution of equations
and up-to techniques arises 
in the (asynchronous)  $\pi$-calculus. 
In this setting, forms of bisimulation enhancements that involve `up
to context',
\iflong
 such
as `up to expansion and context', 
\fi
 require closure of the candidate relation under
substitutions 
(or instantiation of the
 formal parameters of an abstraction with
arbitrary actual parameters).
It is an
open problem whether this closure is necessary 
in the asynchronous
$\pi$-calculus, where
 bisimilarity 
 is closed under substitutions. 
Our unique-solution techniques
    are strongly reminiscent of up to
context techniques (the body of an equation acts like a context that
is erased in a proof using `up to context'); yet, surprisingly, no 
closure under substitutions is required.

As an example of application of our techniques in the $\pi$-calculus,
we revisit the completeness part of the proof of full abstraction for
the encoding of the call-by-name $\lambda$-calculus into the $\pi$-calculus
\cite{cbn,cbn2} with respect to Levy-Longo Trees (LTs). The proof in
\cite{cbn,cbn2} uses `up to expansion and context'.  Such up-to techniques
seem to be essential: without them, it would be hard even to define
the bisimulation candidate.  For our proof using unique-solution, there is one equation for
each node of a given LT, describing the shape of such node.
  (By the time the revision of this paper has 
been completed, we have used to the technique to establish full
abstraction for the encoding of the call-by-value
$\lambda$-calculus~\cite{DBLP:conf/lics/DurierHS18}.)



\emph{Outline of the paper:} Section~\ref{s:back} provides background 
about CCS and behavioural relations. We formulate our main
results for CCS in Section~\ref{s:mt}, and generalise them in an abstract
setting in Section~\ref{s:generalisation}. Section~\ref{s:pi} shows
how our results can be applied to the $\pi$-calculus.

\emph{Comparison with the results published
in~\cite{DBLP:conf/concur/DurierHS17}:} this paper is an extended
version of~\cite{DBLP:conf/concur/DurierHS17}. We provide here
detailed proofs which were either absent or only sketched
in~\cite{DBLP:conf/concur/DurierHS17}, notably for: our two main
unique-solution theorems (Theorems~\ref{thm:roscoe_ccs}
and~\ref{thm:roscoe2_ccs}); the comparison with up-to techniques
in Section~\ref{s:up-to}; the analysis of call-by-name
$\lambda$-calculus in the asynchronous $\pi$-calculus
(Theorem~\ref{t:fullabs}). 
%
%
We also include more detailed discussions along the paper, notably
about the applicability of our techniques (in particular,
Lemma~\ref{l:div_api}, Propositions~\ref{prop:gsos}
and~Proposition~\ref{prop:min}, and Remark~\ref{r:complete}).





\section{Background}
\label{s:back}
\subsubsection*{CCS}
\label{ss:ccs}

We assume  an infinite set  of \emph{names} $  a, b, \ldots $
 and   a  set  of
\emph{constant identifiers} (or simply \emph{constants}) 
 to write 
recursively defined processes. 
The special symbol $\tau$ does not occur in the names and in the 
constants.  
\iflong
The class   of the  CCS processes  is
built from the operators
of  parallel
composition, guarded sum,
 restriction, 
and  constants, and the guard of a sum 
can be an input, an output, or a   silent
prefix:
\else
We recall  the grammar of CCS:
\fi
$$ 
\begin{array}{ccl}
P  & := &   \;  P_1 |  P_2 \; \midd  \;
\Sigma_{i\in I} \mu_i.   P_i\; \midd 
  \res a P  \;  \midd \;   K   
   \end{array}
\qquad\qquad
\begin{array}{ccl}
\mu    & := &  a    \; \midd  \; \outC a 
  \; \midd       \;  \tau   
   \end{array}
 $$
where $I$ is a countable indexing set.
\iflong
Sums are guarded to ensure that behavioural equivalences and
preorders are substitutive. 
\fi 
We write $\nil $ when $I$ is
empty,
and $P + Q$  for binary sums,
\iflong
 with the understanding
 that, to fit the above grammar, $P$ and $Q$ should be sums of
 prefixed terms. 
\fi
Each constant $K$ has  a definition 
 $K \Defi P$.
We sometimes omit trailing $\nil$, e.g., writing $a|b$ for $a.\nil |b
.\nil $ . 
\iflong
We write
$\mu^n.P$ for $P$ preceded by $n$ $\mu$-prefixes.
In a few examples we write $! \mu. P$ as abbreviation for the constant
$K_{\mu.P} \Defi \mu. (P| K_{\mu.P} )$.
\fi
The operational semantics is given by means of a Labelled Transition System (LTS), and is
given in Figure~\ref{f:LTSCCS} (the symmetric versions  of the rules \ttt{parL} and 
\ttt{comL} have been omitted). 
\iflong
The \emph{immediate derivatives} of a
process $P$ are the elements of  the set $\{P' \st P \arr\mu P' \mbox{ for some $\mu$ }
\}$.   
We use $\ell$ to range over
 visible actions (i.e.,  inputs or outputs, excluding  $\tau$).
\fi
\begin{figure*}
\begin{center}
\ttt{sum} $ \displaystyle{\over \Sigma_{i\in I} \mu_i.   P_i 
\arr{ \mu_i  }P_i } $  $ \hb$ 
$\; \;$ \ttt{parL} $\displaystyle{   P \arr\mu   P' \over   P | Q   \arr\mu
P' | Q } $  $ \hb$   
 $\; \;$ \ttt{comL} $\displaystyle{ P \arr{ a}P' \hk \hk  Q
\arr{\outC a }Q'  \over     P|  Q \arr{ \tau} P'
|  Q'  }$ 
\\
\ttt{res} $\displaystyle{ P \arr{\mu}P' \over
 \res a     P   \arr{\mu} \res a P'} $ $ \mu \neq a, \outC a$
$ \hb$ 
\hskip .5cm 
\const$\displaystyle{ P \arr{\mu}P' \over
 K   \arr{ \mu} P'  } $  if  $  K \Defi P$
\end{center}
\caption{The LTS for CCS}
\label{f:LTSCCS}
\end{figure*}

Some standard notations for transitions:  $\Longrightarrow $ is the 
reflexive and  transitive closure of $\arr\tau $, and 
$\Arr \mu $ is $\Longrightarrow \arr\mu \Longrightarrow $ (the
composition of the three relations).
Moreover,   
$ 
P \arr{\hat \mu} P'$ holds if $P \arr\mu P'$ or ($\mu =\tau$ and
$P=P'$); similarly 
$ 
P \Arcap \mu P'$ holds if $P \Arr\mu P'$ or ($\mu =\tau$ and
$P=P'$).
\iflong
We write $P (\arr\mu)^n P'$ if $P$ can become $P'$ after performing
$n$ $\mu$-transitions. Finally, $P \arr\mu$ holds if there is $P'$
with $P \arr\mu P'$, and similarly for other forms of transitions.

\subsubsection*{Further notations.}
\fi
Letters  $\R,\S$ range over relations.
We use the infix notation for relations, e.g., 
$P \RR Q$ means that $(P,Q) \in \R$, and we write
 $\R\S$  for the composition of $\R$ and
  $\S$. A relation \emph{terminates} if there is no infinite sequence
  $P_1\R P_2\R\dots$.
We use a tilde to denote a tuple, with countably many elements; thus
the tuple may also be infinite.
 All
notations  are  extended to tuples componentwise;
e.g., $\til P \RR \til Q$ means that $P_i \RR Q_i$, for  each  
component $i$  of the tuples $\til P$ and $\til Q$.
We use  
\iflong
symbol 
\fi
$\defi$ for abbreviations; 
\iflong
for instance, $P \defi G $, where
$G$ is some expression, means that  $P$ stands
for the  expression
$G$. 
\fi In contrast,  
\iflong
symbol
\fi $\Defi$ is used for the definition of
constants, and   $=$  for syntactic equality and for equations.
\iflong
If $\leq$ is a preorder, then  $\geq$  is its inverse (and
conversely).
\fi

We focus on \emph{weak} behavioural relations, which
abstract from the number of internal steps performed 
\iflong
 by equivalent
processes, because they are, in  practice, the important relations.  
\fi             
\begin{definition}[Bisimilarity]
\label{d:wb}
A 
 relation ${\R}$ 
 is a \emph{bisimulation} if, whenever
 $P\RR Q$, 
we have:
\begin{enumerate}
\item 
    $P \arr\mu P'$   implies that there is $Q'$ such that $Q \Arr{\hat
\mu} Q'$ and $P' 
\RR Q'$;

\item the converse,
 on the actions from $Q$.
\end{enumerate}  
 $P$ and $Q$ are \emph{bisimilar},
written $P \wb
 Q$, if $P \RR Q$ for some  bisimulation $\R$.  
\end{definition}

\subsubsection*{Systems of  equations.}
\label{ss:SysEq}

             

We need variables to write equations. We  use
 capital
letters  $X,Y,Z$
 for  these variables and call them \emph{\behav\  variables}.
 The body of an equation is a CCS expression
possibly containing \behav\  variables. 
%
We use $E, E',\dots$ to range over \emph{\ees}; these
are process expressions that may contain occurrences of variables;
that is, the grammar for processes is extended with a production for
variables.

\begin{definition}
Assume that, for each $i$ of 
 a countable indexing set $I$, we have a variable $X_i$, and an \See
$E_i$, possibly containing  some variables. 
Then 
$\{  X_i = E_i\}_{i\in I}$
(sometimes written $\til X = \til E$)
is 
  a \emph{system of equations}. (There is one equation for each variable $X_i$.)
\end{definition}

\iflong
  In the equation $X=E[X]$, we sometimes call body of the equation
  the \ee $E$ (we use the same terminology for systems of equations).

\fi  

$E[\til P]$  is the process resulting from $E$ by
replacing each variable $X_i$   with the process $P_i$, assuming
$\til P$ and $\til X$ have the same length. (This is syntactic
replacement.)
The components of $\til P$ need not be
 different from each other, while this 
must hold for the variables $\til
 X$. 

\begin{definition}\label{d:un_sol}
Suppose  $\{  X_i = E_i\}_{i\in I}$ is a system of equations. We say that:
\begin{itemize}
\item
 $\til P$ is a \emph{solution of the 
system of equations  for $\wb$} 
if for each $i$ it holds
that $P_i \wb E_i [\til P]$.

\item The  system has 
\emph{a unique solution for $\wb$}  if whenever 
 $\til P$ and $\til Q$ are both solutions for $\wb$, then $\til P \wb \til Q$. \end{itemize}
 \end{definition}

\iflong
Examples of systems with a  unique solution for $\wb$ are: 
\begin{enumerate}
\item
$ X = a. X$ 

\item 
$ X_1 = a.  X_2$, $ X_2 = b.  X_1$  

\end{enumerate}
 
The unique solution of the system (1), modulo $\wb$,  is the constant $K \Defi a
. K$:  for any other solution $P$ we have $P \wb K$.
The unique solution of (2), modulo $\wb$, is given by the constants $K_1 , K_2$
with $K_1 \Defi a . K_2$ and $K_2 \Defi b. K_1$; again, for any other
pair of solutions $P_1,P_2$ we have $K_1 \wb P_1$ and $K_2 \wb P_2$.

Examples of systems that do not have unique solution are: 
\begin{enumerate}
\item $X = X $ 

\item $X = \tau . X$
\item $X = a | X$

\end{enumerate} 
All processes are solutions of (1) and (2); examples of solutions for
(3) are $K$ and $K | b$, for $K \Defi a
. K$
\else
For instance, 
the system $ X_1 = a.  X_2$, $ X_2 = b.  X_1$  has a unique solution,
whereas the equations 
 $X = X $, or  
 $X = \tau . X$, or
 $X = a | X$ do not. 
\fi
{\begin{remark}
To prove that two processes $P$ and $Q$ are equivalent using the unique solution  
proof technique, one has first to find an equation $X=E[X]$ of which both $P$ and 
$Q$ are solutions. Then, a sufficient condition for uniqueness of
solutions makes it possible to deduce 
$P\bsim Q$. 

  The unique-solution method is currently particularly
used in combination with algebraic laws~\cite{RosUnder10,BaeBOOK,groote}.


\end{remark}}
\iflong
\begin{definition}
An \ee   $E$
 is 
\begin{itemize}
\item
\emph{strongly guarded} if each occurrence of
a variable in $E$ is underneath a visible prefix;

\item
\emph{(weakly) guarded} if each occurrence of
a variable in $E$ is underneath a prefix, visible or not;

\item 
 \emph{sequential}  if each occurrence of
a variable in $E$ is only appears  underneath prefixes and sums .
\end{itemize}
We say that an equation satisfies one of the above properties when its body does.
These notions are extended to systems of equations in a natural way:
for instance, $\{  X_i = E_i\}_{i\in I}$ is guarded if each expression
$E_i$ is (w.r.t.\ every variable that occurs in $E_i$).
 \end{definition}
\else

A system of equations 
is \emph{guarded} 
(resp.\ \emph{strongly guarded})
if each occurrence of
a variable in the body of an equation
 is underneath a prefix (resp.\ a visible prefix, i.e., different from
 $\tau$).
 \iflong
 In the literature, guarded is sometimes called
 \emph{weakly} guarded.\fi
\fi
\iflong
In other words,  
if the system is sequential, then 
 for
every expression $E_i$, any sub-expression of $E_i$ in which $X_j $ 
appears, apart from $X_j$ itself,  is a sum (of prefixed terms). 
For instance, 
\begin{itemize}
\item $X = \tau. X + \mu . \nil$ is sequential but not 
 strongly guarded, because the guarding prefix for the variable
is not visible.

\item $X =  \ell . X | P$ is  guarded but not sequential.

\item $X =  \ell . X + \tau. \res a (a .\outC b | a.\nil)$, as well as 
$X = \tau . (a. X + \tau . b .X + \tau  )$
are both 
 guarded and sequential.
\end{itemize} 

{
\begin{remark}[Recursive specifications in ACP]
Systems of equations
are called \emph{recursive specifications}
in the literature related to ACP~\cite{acp}.
%
%
In that context, other notions of guardedness have been studied.
A sufficient condition, more
powerful than Milner's,  was originally
given by Baeten, Bergstra and Klop~\cite{acp}, where synchronisation
is transformed into a visible action, which is then deleted through an
explicit operator. An equation is then guarded if no such deletion
operator appears in its body.

In some process calculi, such as ACP~\cite{acp} and mCRL2~\cite{groote}, guardedness is synonymous, for an 
equation, to having a unique solution: this follows the 
\emph{Recursive specification principle} (RSP), which states that guarded recursive 
specifications are unique (while the \emph{Recursive definition principle} states that 
recursive specifications do have solutions). 
Increasingly general definitions of 
guardedness that make this principle sound are then studied.
 In this framework, 
our contribution can be seen as a new notion of guardedness, stronger that what is found 
in the literature, under which the recursive definition principle is sound for weak bisimilarity.
\end{remark}
}

\begin{theorem}[unique solution of equations, \cite{Mil89}]
\label{t:Mil89}
A system of strongly guarded and sequential equations
   has 
a unique solution
 for $\wb$.
\end{theorem} 

The proof exploits an invariance property on immediate transitions for
strongly guarded and sequential expressions, and then extracts a bisimulation
(up to bisimilarity) out
of the solutions of the system.  

To see the  need of the
 sequentiality  condition, 
  consider
 the equation (from \cite{Mil89}) 
 \[ X = \res a (a. X | \outC a)
 \enspace,\]
 where the occurrence of $X$ in the equation expression is strongly
 guarded but
 not sequential. Any process that does not use $a$ is a solution.

\subsubsection*{Note.} In CCS, $\nu$ is not a binder; 
this allows us to write equations where local names are used
outside their scopes (for instance, $ X=a.\res {a}(\outC a | X)$), as
there is no alpha-renaming. It would still be possible, when $\nu$ has
to be considered a binder (in a higher-order setting, for example), to
write similar equations in a parametric fashion: $X=(a)~a.\res
{a}(\outC a | X\langle a\rangle)$. Such an approach is adopted in
Section~\ref{s:pi}, to handle name passing.

\fi


\begin{definition}[Divergence]
  A process $P$ \emph{diverges} if it can perform an infinite sequence
  of internal moves, possibly after some visible ones; i.e., there are
  processes $P_i$, $i\geq 0$, and some $n$, such that
  $P=P_0\arr{\mu_0} P_1 \arr{\mu_1} P_2 \arr{\mu_2}\dots$ and for all
  $i>n$, $\mu_i=\tau$. We call a \emph{divergence of $P$} the sequence
  of transitions $\big(P_i\arr{\mu_i}P_{i+1}\big)_{i}$\enspace.
\end{definition} 

\begin{example}
  The process $ L\Defi a.\res {a}(L|\outC a)$ diverges, since
  $L\arr a \res a (L|\outC a)$, and (leaving aside $\nil$ and useless
  restrictions) $\res a (L|\outC a)$ has a $\tau$ transition onto
  itself.
\end{example}

\section{Main Results}
\label{s:mt}
\subsection{Divergences and Unique Solution}
This section is devoted to our main results for %
bisimilarity, in the case of CCS.  
%
We need to reason with the unfoldings of the given  equation $X = E$: 
{ we define} the $n$-th unfolding of $E$ {to be}  $E^n$; 
thus  $E^1$ { is defined as} $E$,   
  $E^2$ {as} $E[E]$,   and 
  $E^{n+1}$ { as} $E^n[E]$. 
The \emph{infinite} unfolding
represents the simplest
and most intuitive solution to the equation. In the CCS grammar, such a solution is
obtained by turning the equation into a constant definition, namely
the constant $K_E$ with $K_E \Defi E[K_E]$. We call $K_E$ the 
\emph{syntactic solution of the equation}.

For  a system of equations  $\til X=\til
E[\til X]$, the unfoldings are defined accordingly 
(where $E_i$ replaces $X_i$ in the unfolding): 
 we write
$\til{E}^2$  for the system $\{X_i=E_i\croc{ \til E}\}_{i\in I}$, and similarly for
$\til{E}^n$. 
The syntactic solutions
are defined to be the set of mutually recursive
 constants $\{\ssoli E i\Defi E_i[\ssols E]\}_i$.

As \ees are terms of an extended CCS grammar, that includes variables,
we can apply the SOS rules of CCS to them (assuming that 
variables have no
transitions). 
We extend accordingly to \Sees notations and terminology for LTS and
transitions.

We have the following properties for transitions of processes of the
form $E[\tP]$, where the transition emanates from the $E$ component only:
\begin{lemma}[Expression transitions]\label{l:context:lts}~
  \begin{enumerate}
  \item\label{it:l:exptrans} Given $E$ and $E'$ two \ees, if $E\arr\mu E'$, then
    $ E[\tP]\arr\mu E'[\tP]$, for all processes $\tP$, and
    $E\croc {\tF}\arr\mu E'\croc{\tF}$ for all \ees $\tF$.
\item If $E$ is a guarded \See and $E[\tP]\arr \mu T$, then there
  is an \See $E'$ such that $E\arr\mu E'$ and $T=E'[\tP]$. Similarly for a transition
  $E[\tF]\arr{\mu}E'$.
  \end{enumerate}
\end{lemma}

\begin{proof}
\begin{enumerate}
\item By a simple induction on the derivation of the transition.
\item By a simple induction on the \ee $E$.
\end{enumerate}
\end{proof}

In the hypothesis of case~\ref{it:l:exptrans} above, 
  we sometimes call a transition 
$ E[\tP]\arr\mu E'[\tP]$ 
 \emph{an
    instance} of the \See transition  $E\arr\mu E'$.

\begin{definition}[Reducts]
\begin{enumerate}
\item The \emph{set of reducts} of an  \See $E$, written $\red E$, is
given by: 
$$\red E\quad\Defi\quad  \bigcup_n  \: \{E_n \st 
E\arr{\mu_1}E_1 \dots\arr{\mu_n} E_n \mbox{ for some $\mu_i,E_i$ ($1\leq i\leq n$) }\}
.$$
%
\item The set of \emph{reducts of the unfoldings} of a system of equations $\{X_i=E_i[\til X]\}_{i\in I}$, 
also written $\reds {\til E}$, is defined as 
$$\reds {\til E}\quad \Defi \quad\bigcup _{n\in\N,i\in I}\red{E_i^n}$$
\end{enumerate}
\end{definition}
%
%
\begin{definition}
A system of equations $\til E$ 
 \emph{protects its solutions} if, for all 
solution $\til P$ of the equation $\til X=\til E[\til X]$, the following holds: 
for all  $E'\in \reds {\til E}$, if
 $E'[\til P]\xRightarrow{\mu}Q$  for some $\mu$ and $Q$,
then there exists $E''$ and $n$ such that
 $E'\croc{\til E^n}\xRightarrow{\hat\mu} E''$, 
 and
 $E''[\til P] \bsim Q$.
\end{definition}
Consider a single equation $X=E[X]$: it protects its solutions when all sequences of
transitions emanating from $E'[P]$, where $P$ is a solution of $E$ and
$E'$ is a reduct of the unfoldings of $E$, can
be mimicked by transitions involving unfoldings of $E$ and reducts of
$E$ only (without $P$ performing a transition).
  This is a technical condition, which is useful in the proofs.  
  In this paper, all examples of equations having a unique solution satisfy this property.
\begin{proposition}
\label{l:protects}
A system of equations
that protects its solutions 
  has a unique solution for $\bsim$.
\end{proposition}
\begin{proof}
  Given two solutions $\til P$, $\til Q$ of the system of equations
   $\til X = \til E[\til X]$ 
  we prove that the relation
  $$\R\quad\Defi\quad \{(S,T)~|~ \exists E',\text{ s.t. }S\bsim\C {\til P},\,
  T\bsim\C {\til Q}\text{ and } E'\in \reds {\til E}\}$$
  is a bisimulation relation such that $\til P\R\til Q$. 

  We consider $(S,T)\in \R$, that is, 
  $S\bsim \C \tP$ and
  $T\bsim \C \tQ$, for some $E'\in\reds \tE$. If $S\arr{\mu}S'$, then by
  bisimilarity $\C \tP\xRightarrow{\hat\mu}S''\bsim S'$.
  Since $\tE$ protects its solutions, there are $n,~ {E''}$ such that
  $ E'\croc {\tE^n}\xRightarrow{\hat\mu}E''$ and $\Cp \tP\bsim
  S''$. 
  We can deduce by Lemma~\ref{l:context:lts}(\ref{it:l:exptrans}) that 
  $E'\croc {\tE^n [\tQ]}\xRightarrow{\hat\mu} {E''}[\tQ]$. Since $\tQ$ is a
  solution of $\tE$, we have $\tQ\bsim \tE^n[\tQ]$. This entails, as we know by
  hypothesis
  $T\bsim E'[\tQ]$, that 
  $T\bsim\C \tQ\bsim E' \croc{\tE^n [\tQ]}$. By bisimilarity, we deduce from
  $E'\croc {\tE^n [\tQ]}\xRightarrow{\hat\mu} {E''}[\tQ]$
  that
  $T\xRightarrow{\hat\mu} T'\bsim E'' [\tQ]$.

  The situation can be depicted on the following diagram:

\centerline{
\xymatrix @M=0.4pc @C=0pc 
{ S \ar@{->}[d]_{\mu} & \bsim & \C \tP \ar@{=>}[d]_{\hat\mu}
 & \bsim & E' \croc{\tE^n[\tP]}\ar@{=>}[d]_{\hat\mu}
 \ar @{} [drr] |{=}& & E'\croc{\tE^n[\tQ]}\ar@{=>}[d]_{\hat\mu}
 & \bsim & \C \tQ
 & \bsim & T\ar@{=>}[d]_{\hat\mu}
\\
S' & \bsim & S'' & \bsim & \Cp \tP & & \Cp \tQ & & \bsim & & T'
} }

Then, from $E'\in\reds \tE$ and $E'\croc {\tE^n}\Arr{\hat\mu}E''$, we
deduce that $E''\in\reds \tE$. Finally, $T'\bsim\Cp \tQ$ and $S'\bsim\Cp
\tP$ give that $(S,T)\in\R$.

We reason symmetrically for $T\xrightarrow\mu T'$.
\end{proof}



We say that the syntactic solutions of the system $\til E$ do not diverge if, for all $i\in I$, $\ssoli E i$ does not diverge.

\begin{theorem}[Unique solution]
\label{thm:roscoe_ccs}
A guarded system of equations 
whose  syntactic 
solutions 
 do not diverge has a unique solution for $\bsim$.
\end{theorem}

\begin{proof}
To enhance readability, we first give the proof for a single equation $X=E[X]$. 
We discuss the generalisation to a system of equations at the end of the proof.

  Given a guarded \ee $E$, we prove that 
  $E$ protects the
solutions of its equations. 
To prove that, 
we need to consider a transition $E_0 [P]\xRightarrow{\mu} P'$, for
some $E_0\in \reds E$ and some solution $P$ of $E$.

We build a sequence of expressions $E_n$, and an increasing
sequence of transitions $ E_0 \croc{ E^n} \xRightarrow{\mu?}E_n$
(where $\xRightarrow{\mu?}\defi\Rightarrow\cup\xRightarrow{\hat\mu}$)
such that: either this construction stops, yielding a transition
$E_0\croc{E^n}\Arr{\mu}E_n$, or the construction is
infinite, therefore giving a divergence of $\ssol E$.

We build this sequence so that it additionally satisfies:
\begin{itemize}
\item Either we have
 $E_0 \croc {E^n} \xRightarrow{\hat\mu} E_n $ and $ E_n[P]\Rightarrow
  \bsim P'$,
 or $E_0 \croc{ E^n } \Rightarrow E_n $ and $E_n[P]\xRightarrow{\hat\mu}
  \bsim P'$.
\item The sequence is strictly increasing: the sequence of transitions
$E_0\croc{E^{n+1}}\Arr{\mu?}E_n[E]$ is a strict prefix of the sequence of transitions
$E_0\croc{E^{n+1}}\Arr{\mu?}E_{n+1}$
\end{itemize}
\iflong
This construction is illustrated by Figure~\ref{f:div_proof_ccs}. 
\fi
We start with the empty sequence from $E_0$. 

\medskip

Suppose therefore that at step $n$,  we have for example
$E_0 \croc {E^n[P]} \xRightarrow{\mu}E_n[P]\Rightarrow T_n$, with
$P'\bsim T_n$.
\begin{figure}[t]
\centerline{
\xymatrix @M=0.4pc @C=0pc @R=1.5pc 
{\ar @{} [drr] |{=}
 E_0 \croc{ E^n[P] } \ar@{=>}[d]_{\mu ?} & &  E_0 \croc{ E^{n+1}[P]}
 \ar@{=>}[d]_{\mu ?} & & = & E_0\croc{ E^{n+1}[P]}
 \ar@{=>}[dd]_{\mu ?}\\ 
 E_n[P]\ar@{=>}[dd]_{\mu ?} & \bsim &
 E_n[E[P]]\ar@{=>}[dd]_{\mu ?}  & = & E_n \croc {E[P]}
 \ar@{=>}_-{\mu ?} [d]  & *{} \\ 
*{} & *{}  & \ar@{..}[rr]  & *{} & *{} \ar@{..}[r] \ar@{=>}_-{\mu ?} [d]  & E_{n+1}[P]  \\
\save[] *[l]{P'~~\bsim ~~ T_n}\restore {\color{white}A}& \bsim & T_{n+1} & = & T_{n+1} &
} }
\caption{Recursion: construction of the sequence of transitions $E_0\croc{E^n[P]}\Arr{\mu?}E_n$}
\label{f:rec_proof_ccs}
\end{figure}
\iflong
\begin{figure}
\centerline{ \xy    
\xymatrix "M" @M=0.65pc @C=0pc 
{
  E_0 [P] \ar@{=>}[ddddd]_{\hat\mu }
& \bsim &\textcolor{cyan}{  E_0 [ E }[P]\textcolor{cyan}{]}\ar@{=>}[ddddd]_{\hat\mu }
& \bsim & \textcolor{cyan}{ E_0[ E^2} [P]\textcolor{cyan}{]}{\ar@{=>}@[cyan][d]_{\textcolor{cyan}{\mu? }}}
& \bsim &\textcolor{cyan}{  E_0 [ E^3}[P]\textcolor{cyan}{]}\ar@{=>}@[cyan][dd]_{\textcolor{cyan}{\mu? }}
& \bsim & \dots 
& \bsim  & \textcolor{cyan}{  E_0[ E^n}[P]\textcolor{cyan}{]}\ar@{=>}@[cyan][ddddd]_{\textcolor{cyan}{\mu? }}\\
& & *{}\ar@{..}[rrr]  & & *{}\ar @{=>}_{\mu?}[dddd] & & &  & & \\
& & & & *{}\ar@{..}[rrr]& &*{} \ar@{=>}[ddd]_{\mu?}& & &  \\
& & & &  & & & & & &\\
& & & & & & & & & &  \\
P'&\bsim & T_1 &\bsim & T_2 &\bsim & T_3 &\bsim &\dots &\bsim & \textcolor{cyan}{ {E_n} }[P]
 }\textcolor{cyan}{
  \POS"M2,5"."M3,5"!C*++L-U\frm{\{},
    \POS"M3,7"."M4,7"!C*+L+U-U\frm{\{},
      \POS"M1,3"."M2,3"+<-1ex,-7ex>*+<0ex,6ex>\frm{\{}}
 \endxy }
\caption{Construction of the \See transition}
 \label{f:div_proof_ccs}
\end{figure}
\fi
%
\begin{itemize}
\item If $E_n[P]\Rightarrow T_n$ is the empty sequence, we stop. We
  have in this case $E_0 \croc E^n \xRightarrow{\mu} E_n$ and
  $P'\bsim E_n[P]$. 
\item Otherwise (as depicted on Figure \ref{f:rec_proof_ccs}), we
  unfold further equation $E$: we have 
  $E_0\croc{E^{n+1}}\xRightarrow{\mu} E_n\croc E$ by Lemma
  \ref{l:context:lts}. By congruence and bisimilarity we have
  $E_n\croc {E [P]} \Rightarrow T_{n+1}\bsim P'$ for some
  $T_{n+1}$. If $E_n\croc {E [P]} \Rightarrow T_{n+1}$ is the empty
  sequence, we stop as previously. Otherwise, we take
  $E_n\croc E\Arr{\mu?}E_{n+1}$ to be the longest prefix sequence of
  transitions in $E_n\croc {E [P]} \Rightarrow T_{n+1}$ that are
  instances of \See transitions from $E_n\croc E$ (we remark that as $E_n\croc E$ is
  guarded, this sequence is not empty).
\end{itemize}

Suppose now that the construction given above never stops. 
We know that
$E_n\croc{E} \Arr{\mu?} E_{n+1}$, therefore
$E_n[\ssol E]\Arr{\mu?}E_{n+1}[\ssol E]$. This gives an infinite
sequence of transitions starting from $E_0[\ssol E]$:
$E_0[\ssol E]\Arr{\mu?}E_1[\ssol E]\Arr{\mu?}\dots$. We observe that
in the latter sequence, every step involves at least one transition,
and moreover, there is 
at most one visible action ($\mu$) occurring in this infinite
sequence.  Therefore $E_0[\ssol E]$ is divergent, which contradicts
the hypothesis of the theorem. Hence, the construction does stop, and
this concludes the proof.

\subsubsection*{Systems of equations.} 
To extend the previous proof 
to systems 
of equations, we consider solutions $\tP$ to the system $\til E$, and 
use unfoldings of the system of equations (instead of unfoldings of a single equation). 
We also reason about an initial transition from $E_0[\tP]$, where $E_0$ is a reduct of the unfoldings 
of one of the equations (i.e., $E_0\in\reds {\til E}$).
\end{proof}

\iflong
\begin{figure}[h]
\centerline{
\xymatrix @M=0.4pc @C=0pc @R=1.5pc 
{P \ar[ddd]_\mu &\bsim & E[P]\ar@{=>}[ddd]_{\hat\mu} & \bsim & E\croc{E[P]}\ar@{=>}[ddd]_{\hat\mu} & 
 \bsim & \dots &\bsim & E^n[P]\ar@{=>}[ddd]_{\hat\mu}& & E^n[\cdot]\ar@{=>}[ddd]_{\hat\mu}&& E^n[Q]\ar@{=>}[ddd]_{\hat\mu}&\bsim &Q\ar@{=>}[ddd]_{\hat\mu}\\
&&*{}\ar@{..}[rrrrrrrr] &&&&&&&&*{}&&*{}&&\\
&&&&*{}\ar@{..}[rrrrrr]&&&&&&*{}&&*{}&&\\
*{}&&*{}&&*{}&&&&*{}\ar@{..}[rrrr]&&*{}&&*{}&&*{}}}
\caption{Proof of Theorem \ref{thm:roscoe_ccs}, building a a
  transition of $E^n[Q]$ where $Q$ does not move}
\label{f:proof_ccs}
\end{figure}
\fi

\subsection{Innocuous Divergences}\label{s:innocuous}

 In the remainder of the section we refine Theorem~\ref{thm:roscoe_ccs}
 by taking into account only certain forms of divergence. 
To introduce the idea, 
consider 
 the equation $X=  a. X | K $,  for $K \Defi \tau. K$:
 the divergences  induced by
$K$
do not prevent uniqueness of the solution, as any  solution $P$
necessarily satisfies $P\bsim a.P$. 
Indeed the variable of the equation is strongly 
guarded and  a visible action has to be produced
before accessing the variable. 
These divergences are not dangerous because they
 do not percolate through the infinite unfolding of the equation; in other
words, { a finite unfolding may produce the same divergence, therefore } it is not necessary to go to the infinite unfolding to diverge.
We call such divergences \emph{innocuous}. Formally, these divergences
are derived by applying only a finite number of times  rule 
\const of the LTS 
 (see Figure~\ref{f:LTSCCS}) to the constant
that represents the syntactic solution of the equation.
Those divergences 
can be understood as instances, in the syntactic solution, of divergences 
of some finite unfolding of the equation: 
{Consider a (single) equation  $X=E[X]$; this way, any divergence of $E^n$ can be 
 transformed into an innocuous divergence of $K_E$. This idea is also
behind Lemma~\ref{l:inj} below.
}



This refinement fits well the operational approach we adopt to
 formalise the results; 
it looks less 
 natural in a  denotational or trace-based 
setting   for CSP like  in \cite{roscoe1,roscoe2}, where  any divergence causes a process to
be considered undefined.

\begin{definition}[Innocuous divergence] 
  Consider a guarded system of equations $\til{X}= \til{E}$ and its
  syntactic solutions $\ssols E$.  A
  divergence 
  of $\ssoli E i$ (for some $i$) is called \emph{innocuous} when, 
  summing up all usages of rule \const with one of the $\ssoli E j$s
  (including $j=i$) \emph{in all derivation proofs of the transitions
   belonging to the divergence}, we obtain a finite number.
\end{definition}

\begin{theorem}[Unique solution with innocuous divergences]
\label{thm:roscoe2_ccs}
Let $\til{X}=\til{E}$ be a system of guarded equations, and $\ssols E$
be its syntactic solutions. If all divergences of any $\ssoli E i$ are
innocuous,
then $ \til{E}$ has a unique solution for $\bsim$.
\end{theorem}

\begin{proof}
  We reason like in the proof of Theorem~\ref{thm:roscoe_ccs}.  We
  explain the difference, in the case where we have a single
  equation  
  (bearing in mind that the
  generalisation discussed at the end of the proof of
  Theorem~\ref{thm:roscoe_ccs} can be carried over accordingly).

Along the
construction in that proof,
if at some point the transition $E_0 \croc{ E^n[P]}
\xRightarrow{\hat\mu}T_{n} $
 is an instance of an \See transition $E_0 \croc{ E^n}
\xRightarrow{\hat\mu} E_n $, 
 then the construction can stop.

 Consequently, if the construction never stops, we build a
 non innocuous divergence: we can assume that for any $n$,
 $E_0\croc {E^n [P]}\xRightarrow{\hat\mu}T_n$ is not an instance of an
 \See transition from $E_0\croc {E^n}$. This means that in this
 sequence of transitions, the LTS rule \const is used 
 at least $n$ times applied to the constant $\ssol E$.  
 Hence,  the divergence we build uses at least $n$ times
 the rule \const applied to the constant $\ssol E$, for all $n$.
 Therefore, 
 the divergence is not innocuous, which is a
 contradiction. 
\end{proof}

\begin{remark}\label{r:altguard}
The conditions for unique solution in Theorems~\ref{thm:roscoe_ccs} and~\ref{thm:roscoe2_ccs} 
combine 
syntactic 
(guardedness)
and semantic  (divergence-free) conditions.
A purely semantic condition can be used if 
rule 
\const of Figure~\ref{f:LTSCCS} is modified so that the unfolding of a
constant yields a $\tau$-transition: 
$$\overline{K\arr\tau P} ^{\mbox{ if }K\Defi P}$$
Thus in the theorems the condition imposing guardedness of the
equations 
could be dropped.  The resulting theorems
would actually be more powerful because they would
accept equations which are not syntactically guarded: 
it is sufficient that each equation has a finite unfolding which 
is guarded. For instance the system of equations 
$X = b| Y$, $Y = a.X$ would be accepted, although the first
equation is not guarded.
\end{remark}

\begin{remark}\label{r:complete}
Even taking innocuous divergences into account, our criterion for 
equations to have a unique solution is not complete. Indeed, the following 
equation 
$$X=\inpC a .X+\outC a .X+\inpC d .(X|X)$$ 
has a non-innocuous divergence: its syntactic solution, $K$, has the
transition $K\arr d K|K$; then, $K|K$ diverges by unfolding on both
sides at every step. However, the equation has a unique solution,
intuitively because the prefix $d$ acts as a strong guard, that does
not take part in the divergence. This could hint at further
developments of our theory.

Note moreover that this equation is  
non-linear, in the sense that there are occurrences of the same variable $X$ in 
parallel. In practice, such examples are not very useful;  for linear
equations, we were 
not able to find  counter-examples.
\end{remark}

\subsection{Conditions for unique solution}

The following lemma states a condition  to ensure that all divergences
produced by a system of equations are innocuous.
{This condition is decidable, but weak. However, in practice, it is 
often satisfied;} 
it is in particular sufficient 
in all examples in the paper.

\begin{lemma}
\label{l:criterion:div}
Consider a system of equations $\til{X} =\til{E}$, and suppose that for each $i$ there is
$n_i$ such that in $E_i^{n_i}$, each variable is underneath a visible
prefix (say, $a$ or $\outC a$) whose complementary prefix ($\outC a$
or $a$) 
does not appear in any equation. Then the system has only innocuous
divergences.
\end{lemma}

\begin{proof}
The condition ensures that 
at least one of the 
 prefix occurring above the equation variables 
may never take part in a $\tau$ action. This entails that this prefix
cannot be triggered along an infinite sequence of $\tau$ steps. 
Hence a divergence of the unique solution of such a system 
of equations may not require the rule \texttt{const} to be used an 
infinite number of times.
\end{proof}



\subsection{Comparison with other techniques}

We now compare our technique with two alternative approaches: unique
solution of contractions, and enhancements of bisimulation.

\subsubsection{An example (lazy and eager servers), and comparison
  with contractions} 
\label{s:contractions}

We now 
show
an example of application of our technique, 
taken from~\cite{popl15}.
The example also illustrates the relative strengths of the two
unique solution theorems (Theorems~\ref{thm:roscoe_ccs} 
and~\ref{thm:roscoe2_ccs}),
\iflong
 and a few aspects of the comparison
with other bisimulation techniques.  
\fi


For the sake of readability, 
we use a version of CCS with value passing; this could be
translated into pure CCS~\cite{Mil89}.
In a value-passing calculus, $a(x).P$ is an input at $a$ in which $x$ is the
placeholder for the value received, whereas $\out a n.P$ is an output
at $a$ of the value $n$; and  $\app A n $ is a parametrised constant.
This example consists of two implementations of a server; 
this server, when interrogated by
clients at a channel $c$,  should start
 a certain interaction protocol with the client,
 after consulting an auxiliary server
 $A$ at $a$.

We consider the two following implementations of this server:
the first one, $L$, is `lazy', and
consults $A$ only \emph{after} a request from a client has been
received.  In contrast, the other one, $E$, is `eager', and 
 consults $A$ \emph{beforehand},
so  then to be  ready in answering a client: 
$$
\begin{array}{crcl}
\begin{array}{rcl}
L  & \Defi &  c(z). a(x). (L |  \app R {x,z}) \\ 
  E  & \Defi & a(x). c(z). (E |  \app R {x,z} )
\end{array}
  &
~\qquad 
\app A n & \Defi &  \out a n . \app A {n+1} 
\end{array}
 $$ 
Here $ \app R {x,z}$ represents the interaction protocol that is
started with a client, and can  be any process. It may use the
values $x$ and $z$ (obtained from the client and the auxiliary server
$A$);  the interactions produced may indeed depend  on the
values $x$ and $z$. 
We assume for now that  $ \app R {x,z}$ may not use 
channel $c$ and $a$; that is, the interaction protocol that has been
spawned need not come back to the main server or to the auxiliary server.
Moreover we assume $R$ may not diverge.
%
We want to prove that 
the  two servers, when composed with $A$, yield bisimilar
processes. We thus define
$
\app \SL n    \,\Defi\,  \res a (\app A n| L)$ and
$\app \SE n \,  \Defi\,  \res a (\app A n | E)$.
A proof that 
$\app \SL n \wb \app \SE n$
 using the plain bisimulation proof method would be 
long and tedious,  due to the differences between the lazy and the
eager server, and to the fact that $R$ is nearly an arbitrary process. 

The paper~\cite{popl15} presents two proofs of this equivalence. One proof
makes use of the `bisimulation up-to expansion and context' technique;
this makes it possible to carry out a proof using a single pair of
processes for each integer $n$.  A proof of similar size uses the
technique of `unique solution of contractions', by establishing, with
the help of a few simple algebraic laws, that {$\{\app\SL n\}_n$} and
{$\{\app\SE n\}_n$} are solutions of the same system of contractions.

%
We can also build a proof using the technique of `unique solution of
equations'. Milner's original version of this technique cannot be
used, because the equations make use of operators other than just
prefix and sum.
%
In contrast, Theorem~\ref{thm:roscoe_ccs} can be applied.
The  equations are: 
$\{ X_n = \inp c z.    ( X_{n+1}  |  \app R {n,z}) \}_n$.
%
The proofs that the two servers are solutions can be  carried out
using a few  algebraic laws: expansion law, 
structural laws for parallel composition and restriction, 
one $\tau$-law. It is essentially the same proof as that for 
unique solution of contractions.

To apply Theorem~\ref{thm:roscoe_ccs}, we also need to check  that the equations
do not produce divergences. This check is straightforward, as  no silent
move may be produced by interactions along $c$, and any two internal
communications at $a$ are separated by a visible input at $c$.
Moreover, by assumption,  the protocol  $R$ does not produce internal
divergences.

If on the other hand the hypothesis that $R$ may not diverge is lifted, then
Theorem~\ref{thm:roscoe_ccs} is not applicable anymore, and divergences
are possible.  However, such divergences are 
innocuous: the equation need not be unfolded an infinite number
of times for the divergence to occur.
We can therefore still prove the result, by appealing to the more
powerful Theorem~\ref{thm:roscoe2_ccs}.

We can relax the definition of $R$ even further and allow 
calls back to the main server from $R$ itself. 
In this case, interactions between $R$ and the main server (eager or
lazy) may yield divergences (for instance setting $R \defi \outC c$). 
Such divergences need not be innocuous, as intuitively they require
 infinitely many unfolding  of the body of an equation. 
Thus now even Theorem~\ref{thm:roscoe2_ccs} is not applicable. 
(More precisely, for Theorem~\ref{thm:roscoe2_ccs} to fail
$\app R {n,z}$, for each $n$ and $z$, should have 
the possibility of performing an output at $c$ as  first visible
transition.)


This becomes therefore an example in which the `unique solution of
contraction' technique is more powerful, as such technique does not
rely on conditions about divergence and is therefore applicable.



\subsubsection{Comparison with up-to techniques}
\label{s:up-to}

Milner's syntactic condition for unique solution of equations
essentially allows only equations in which variables are underneath
prefixes and sums. The technique is not complete \cite{popl15}; for
instance it cannot handle the server example of Section~\ref{s:contractions}.

The  technique of `unique solution of contractions' 
\cite{popl15} relies on the theory of an
auxiliary preorder (contraction), needed to establish the meaning of
`solution'; and the soundness theorems in \cite{popl15} use a purely
syntactic condition (guarded variables). 
In contrast, our techniques with equations 
 do not  rely on  auxiliary relations and their theory, but the
soundness theorems use a semantic condition (divergence)
, see also
Remark~\ref{r:altguard}).

The two  techniques are incomparable. Considering the server example of
Section~\ref{s:contractions}, the contraction technique is capable of handling
also the case in which the protocol $R$ is freely allowed to make 
 calls back to the main
server, including the possibility that, in doing this, infinitely many
copies of $R$ are spawned. This possibility is disallowed for us, as
it would correspond to a non-innocuous divergence. 
On the other hand, when using contraction, 
a solution is evaluated with respect to the contraction preorder, that
conveys an idea of efficiency (measured against the number of silent
transitions performed). Thus,  while two bisimilar processes are
solutions of exactly the same set of equations, they need not be
solutions of the same  contractions. For instance,
 we can use our techniques to
  prove that processes $K \Defi \tau. a . a . K$ and $H \Defi a. H$
  are bisimilar because  solutions of the equation $X = a.X$;
  in contrast, 
only $H$ is  a solution of the corresponding contraction.

\subsubsection*{Up-to techniques.}
We compare our unique-solution techniques with one of the most
powerful forms of enhancement of the  bisimulation proof method,
namely Pous `up to transitivity and context' technique~\cite{damien_phd}. 
This technique   allows one to use `up to 
weak bisimilarity', `up to transitivity', and `up to context' techniques together.
While `up to weak bisimilarity' and `up to transitivity' are known to be
unsound techniques~\cite{pous:sangiorgi:upto}, here 
they are  combined safely thanks to 
a `control
relation', written below $\succ$, which satisfies a termination hypothesis.
\iflong
This control relation is used to make sure there
is no cyclic or infinite sequence of $\tau$ transitions used to hide other potential
visible transitions. 
{In that regard, this condition is very similar to the non-divergence hypothesis 
of Theorem~\ref{thm:roscoe2_ccs}, as can be seen in the proofs below.} 
\fi

 Formally, contexts are processes that can contain holes,
indexed by some finite set of integers. Therefore, contexts are like \ees,
except that numbered holes are used instead of variable names. In this
section, we will switch freely between \ees and contexts, keeping in
mind that 
contexts
are needed to study up-to techniques while \ees are needed for unique
solution of equations.

We write $\mathcal{C}(\R)$ for the context closure of a relation $\R$,
defined as the set of all pairs $(\ct{\til P},\ct{\til Q} )$ with
$\til P \RR \til Q$. Moreover, $\overline\R$ stands for
$(\bsim\cup\,\mathcal{C}(\R))$, and $\R^+$ for the transitive closure
of $\R$.

\begin{definition}
\label{d:pous:technique}
Let $\succ$  be a  relation that is 
  transitive,  closed under contexts, and such that 
  $\succ(\arr\tau^+)$  terminates. 
 A relation $\R$ is a
  \emph{bisimulation up to  $\succ$ and context} if, whenever $P\RR
  Q$:
\begin{enumerate}
\item  if $P\arr\mu P'$ then
  $Q\Arr{\hat\mu} Q'$ for some $Q'$ with $P'\: \mathrel{(\succ\protect{\cap}\, \overline\R
  )^+ \;\: \mathcal{C}(\R)\; \bsim }\:Q'$;
  
\item the converse on the transitions  from  $Q$.
\end{enumerate}
\end{definition}

\iflong
We refer to~\cite{damien_phd} for more details on this
up-to technique, and for the the proof of its soundness.

We remark that in~\cite{damien_phd}, instead of the transitive closure
$(\succ\cap~\overline{\R})^+$ in Definition~\ref{d:pous:technique}, we
have 
a \emph{reflexive}   and transitive closure.  We do not know if
relaxing this technical condition breaks
Theorem~\ref{thm:upto_equa_complete} below.

\fi

\medskip

We introduce some notations in order to state the correspondence
between the technique introduced in Definition~\ref{d:pous:technique}
and systems of equations.

If $\R $ is a relation, then we can also view $\R$ as an ordered
sequence of pairs (e.g., assuming some lexicographical  ordering). 
Then $\R_i$ indicates the tuple obtained by
projecting the   pairs in $\R$ on the $i$-th component ($i=1,2$).

In the following statement, the \emph{size} of a relation is the
number of pairs it contains, and the size of a system of equations is
the number of equations it consists of.

\begin{theorem}[Completeness with respect to up-to techniques]
\label{thm:upto_equa_complete}
Suppose  $\R$ is   a bisimulation up to 
  $\succ$ and context. Then there exists a 
 guarded system of equations, 
 with only
 innocuous divergences, that admits $\R_1$ and $\R_2$ as solutions.
 Moreover, this system has the same size as $\R$.
\end{theorem}

\begin{proof}

  Suppose $\R=\{(P_i,Q_i)\}_{i\in I}$ is a bisimulation up to $\succ$
  and context.   
  
  We index the transitions of $P_i$ from $1$ to
  $n_i$, 
  and write $P_i\arr{\mu _{i,j}}P'_{i,j}$ for $1\leq j\leq n_i$.

  Then, by Definition~\ref{d:pous:technique}, 
  for $1\leq j\leq n_i$, 
  there exist $C_{i,j}$ and $Q'_{i,j}$ such that we have the following diagram:

\centerline{
\xymatrix @C=0pc {
P_i \ar[d]_{\mu _{i,j}} &&& \R &&& Q_i \ar@{=>}^{\hat{\mu} _{i,j}}[d]\\
P'_{i,j} & (\succ\cap(\bsim\cup\protect{\mathcal{C}(\R)}))^+ & C_{i,j}[\til{P}]
& \mathcal{C}(\R) & C_{i,j}[\til{Q}]&  \bsim & Q'_{i,j}  
}}

We define the system of equations $\til X =\til E$ by setting
$\forall i\in I,~E_i= \sum_{j=1}^{n_i}\mu _{i,j} .C_{i,j}[\til{X}]
$. 
By construction, this system and $\R$ have the same size.

We prove
that 
$\til{P}$
is a solution of $\til
X =\til{E}$. 
We have that $P_i\sim \sum _{j=1}^{n_i}\mu _{i,j}
.P'_{i,j}$. Moreover, by the results in~\cite{damien_phd}, since $\succ(\arr\tau^+)$
terminates, $\R\subseteq\bsim$, which implies 
$(\succ\cap(\bsim\cup\mathcal{C}(\R)))^+\subseteq\protect{\bsim}$. 
Therefore, $P'_{i,j}\bsim
C_{i,j}[\til{P}]$, and $P_i\bsim\sum_{j=1}^{n_i} \mu _{i,j}.
C_{i,j}[\til{P}]$. This shows that $\til{P}$ is a solution of $\til X
=\til{E}$. 

Since $\R\subseteq\bsim$, we have $\til{P}\bsim\til{Q}$, hence
$\til{Q}$ is also a solution of the system of equations.

\medskip

It is left to prove that if for some $i$, $\ssoli E i$ has a
non-innocuous divergence, then $\succ(\arr\tau^+)$ does not terminate.
As this would be contradictory, we will be able to apply
Theorem~\ref{thm:roscoe2_ccs} to finish the proof.


Assume that
$\ssoli E k\arr{\mu_1}\dots\arr{\mu_n}F_0[\ssols E]\arr\tau\arr\tau\dots$,
and that this divergence is not innocuous. 
Based on this divergence, we build a sequence of \ees $F_n$ and $F_n'$ such that: 
$$F_0[\ssols E]\arr\tau^* F'_0[\ssols E]\arr\tau F_1 [\ssols
E]\arr\tau^* F_1'[\ssols E]\arr\tau F_2[\ssols E] \dots$$
In the above divergence, we have that for all $n$, $F_n\arr\tau^* F_n'$ (these are \See transitions) and $F_n'[\ssols E]\arr\tau 
F_{n+1}[\ssols E]$; moreover, the latter transition is not an instance of an \See transition from $F_n'$ 
(meaning that $\ssols E$ contributes to the transition). We will then show 
that $F_n'[\til P]\arr\tau \succ F_{n+1}[\til P]$. Therefore $\succ (\arr\tau^+)$ 
does not terminate.

$F_0$ is already given. 
Assume then that we have $F_n$ such that there is a non innocuous divergence from $F_n[\ssols E]$. 
This entails that there exists $F_n'$ such that
$F_n[\ssols E]\arr\tau^* F_n'[\ssols E]$,
$F_n'[\ssols E]\arr\tau$, and the latter transition is not an instance of 
an \See 
transition from $F_n'$. 
Without loss of generality, we can assume that 
the transitions in $F_n[\ssols E]\arr\tau^* F_n'[\ssols E]$ are 
instances of \See transitions.  

The fact that transition $F_n'[\ssols E]\arr\tau$ is not an instance
of an \See transition means that at least one equation expression
$E_i$ is involved in it. 
For readability, we assume that this transition is of the form
$F_n'[\ssols E]\arr\tau F_n'[(\ssoli E k)_{k<i},Q,(\ssoli E k)_{k>i}]$
for some $k$ and some $Q$,  with $\ssoli E i\arr\tau Q$ (i.e., in
$F_n'$ 
there is only one copy
of $\ssoli E i$, and the transition involves only $\ssoli E
i$).
The reasoning below can be adapted to cases where the $\tau$-transition involves a more complex synchronisation, and/or several copies
of $\ssoli E i$.

Since
$E_i=\sum_j \mu _{i,j}.C_{i,j}[\til X]$, there is $j$ such that
$\tau =\mu _{i,j}$ and $Q=C_{i,j}[\ssols E]$.  We then fix $F_{n+1}$
to be $F_n'\subst{C_{i,j}[\til X]}{X_i}$.  Indeed we have
$F_n'[\ssols E]\arr\tau F_{n+1}[\ssols E]$.

We have observed that $P_i\sim \sum_j \mu _{i,j}.P'_{i,j}$,
and $E_i=\sum_j \mu _{i,j}.C_{i,j}[\til X]$, hence any transition from
$E_i$ can be mimicked by a transition of $P_i$. Therefore we have that
$P_i\arr\tau P'_{i,j}$ 
and
$F_n'[\til P]\arr\tau F_n'[(P_k)_{k<i},P'_{i,j},(P_k)_{k>i}]$.  We
have that $P'_{i,j}\succ^+ C_{i,j}[\til{P}]$ by construction, hence,
since $\succ$ is transitive, $P'_{i,j}\succ C_{i,j}[\til{P}]$;
finally, $\succ$ is closed by contexts, hence
$$F_n'[(P_k)_{k<i},P'_{i,j},(P_k)_{k>i}]\succ F_n'\croc{(P_k)_{k<i}
  ,C_{i,j}[\til{P}],(P_k)_{k>i}} = F_{n+1}[\til P]\enspace.$$

The sequence of \ees $F_n$ and $F_n'$ has thus been defined in such a
way that $F_n\arr\tau^*F_n'$, hence
$F_n[\til P]\arr\tau^* F_n'[\til P]$. We also have
$F_n'[\til P]\arr\tau \succ F_{n+1}[\til P]$, thus 
$F_n[\til P](\arr\tau^+)\succ F_{n+1}[\til P]$. This is an infinite sequence
for $(\arr\tau^+)\succ$, hence we also have non-termination for $\succ(\arr\tau^+)$.  This yields a contradiction.
\end{proof}

Theorem~\ref{thm:roscoe2_ccs} is essential to
establish Theorem~\ref{thm:upto_equa_complete}; 
Theorem~\ref{thm:roscoe_ccs} would be insufficient.

\ifthese
\adrien{Deplacer cette remarque a qui de droit.}
\begin{remark}
The above proof shows how to build a system of equations starting from a pair of processes, 
by relying on the so-called \emph{expansion law}~\cite{Mil89}.
Using this law, one can use guarded sum to express the immediate transitions of
a process. Accordingly, we can turn a process into a 
system of equations, each consisting of 
a guarded sum where the continuations of prefixes are equation
variables.

This form of system of equations
is sometimes called a specification in the literature. Such a specification can then be used to prove bisimilarity results, using 
 unique solution of equations. Milner's unique solution
 theorem (Theorem~\ref{t:Mil89}) can be sufficient in such a situation, as
 long as $\tau$ prefixes do not appear in the specification.
\end{remark}\fi
%
%

 
\section{Abstract Formulation}
\label{s:generalisation}
\subsection{Abstract LTS and \Fcts}
\label{s:lts:ops}

\newcommand{\Tlts}{\ensuremath{\mathcal T}}

In this section we propose generalisations of the unique-solution
 theorems. For this we introduce abstract formulations of them, 
which are meant to highlight their essential ingredients. 
When  instantiated to the
specific case of CCS,
such abstract formulations yield the theorems in Section~\ref{s:mt}.
 The proofs are adapted from those of the corresponding
  theorems in Section~\ref{s:mt}. The results of
this section, up to Theorem \ref{thm:roscoe_abstract}, have been formalised
in Coq theorem prover~\cite{www:usol:coq}.

The abstract formulation is stated  on 
a generic LTS, that is, a triple 
 $\Tlts~=~(S,\Lambda,\rightarrow)$ where: 
 $S$ is the set of states;  $\Lambda$ the set of action labels, containing
the  special label $\tau$
accounting for silent actions;  $\rightarrow$ is the transition
relation. As usual, we write $s_1\arr\mu s_2$
when ${(s_1,\mu,s_2)}\in{\rightarrow}$. 
The definition of weak bisimilarity $\bsim$ 
is
 as in Section~\ref{s:back}. We omit explicit reference to \Tlts{}
 when there is no ambiguity.

We reason about \emph{state functions}, i.e., functions from $S$ to $S$, and use
$f, f', g$ to range over them. 
As usual, $f\circ g$ denotes the composition of $f$ and $g$.
The CCS processes of Section~\ref{s:back} correspond here to the states of
an LTS. In turn, a CCS context $C$ corresponds to 
a {state function}, mapping a process $P$ to the process $C[P]$.

\iflong
\begin{definition}
A state function $f$
  \emph{respects a relation $\rel$} 
   whenever 
$ x\rel y$ implies $ f(x)\rel f(y)$, for all states 
$x,y$.
\end{definition}
\fi

\begin{definition}[Autonomy]
For state functions  $f,f'$ 
we say that there is an \emph{autonomous $\mu$-transition from $f$ to $f'$},
written $f\arr\mu\autonomous f'$, if for all states $x$ it holds that
 $f(x)\arr\mu f'(x)$.

Likewise, given a set $\mathcal{F}$ of 
state functions 
and $f\in
\mathcal{F}$, we say that a transition $f(x)\arr\mu
y$ is \emph{autonomous on $\mathcal{F}$} if, for some  $f'\in
\mathcal{F}$  we have $f\arr\mu\autonomous f'$ and $y=f'(x)$.
Moreover, we say that  \emph{function $f$ is  autonomous
on $\mathcal{F}$}
if all the transitions emanating 
from $f$ 
(that is, all transitions of the form $f(x)\arr\mu y$, for some
$x,\mu,y$)
are autonomous on $\mathcal{F}$.

When  $\mathcal F$
is clear, we omit it, and 
we  simply say that 
a function \emph{is autonomous}.
\end{definition}

Thus,  $f$ is  
autonomous on $\mathcal{F}$ 
if, for some indexing set $I$, 
 there are $\mu_i$ and
\iflong
 functions
\fi $ f_i\in \mathcal{F} $ 
such that  for all  $ x$ 
\iflong
it holds that: 
\else
we have 
\fi
 $ f(x) \arr{\mu_i}  f_i(x)$, for each $i$; 
 the set of all transitions emanating from $f(x)$ is precisely 
$ \cup_i \{ f(x) \arr{\mu_i}  f_i(x)  \}$.
%
\iflong
Autonomous transitions correspond to \See transitions in CCS, 
and
\fi
 autonomous functions correspond to guarded contexts,
which do not need 
\iflong
the contribution of
\fi
 their process argument to
perform the first transition.

We  now formulate conditions under which, intuitively, a state function behaves like a CCS
context. 
Functions satisfying these conditions are called \emph{\fcts}.
%
Such operators are defined relative to a behavioural equivalence or preorder; 
in this section we are interested in bisimilarity, hence the relation $\rel$ in 
the definition below should be understood to be $\bsim$.

\begin{definition}[Set of \fcts]\label{d:opsets}
  Consider an LTS \Tlts, a binary relation $\rel$ on \Tlts, 
  and a set $\mathcal O$ of functions from $S$ to $S$.
We say that  $\mathcal O$ is a \emph{\relsetops R on \Tlts} 
if the following
  conditions hold: 
  \begin{enumerate}
  \item 
$\mathcal{O}$ contains the identity function; 
  \item $\mathcal O$ is closed by composition
(that is, 
$f\circ g\in\mathcal O$ whenever $f, g\in\mathcal O$);
  \item\label{d:operators:autonomy:comp} 
composition preserves autonomy (i.e., 
if $g$ is autonomous on $\mathcal{O}$, then
so is     $f\circ g$); 
  \item  \label{d:operators:congruence}
\iflong
$\rel$ is preserved, 
that is,
\fi
 all functions 
in $\mathcal O$  respect $\rel$
  \end{enumerate}
\end{definition}

A `symmetric variant' of
clause~\ref{d:operators:autonomy:comp} always holds: if $f$ is
autonomous, then so is $f\circ g$: {indeed, transitions of $f(x)$ do not depend on $x$, 
  hence transitions of $f(g(x))$ do not depend on either $x$ or $g$.
  Clause~\ref{d:operators:congruence} expresses congruence of the equivalence 
w.r.t.\ the set of contexts: here, the state functions in $\mathcal{O}$.}
The autonomous transitions of 
a set of \relfcts \rel yield an LTS
whose  states are the \fcts  themselves. Such transitions are of
the form 
$ f \arr\mu\autonomous g$.

In the remainder of this section, we assume $\rel$ to be equal to bisimilarity, 
and call \emph{set of operators} any \relsetops \bsim, without specifying 
the relation. 
Where the underlying set
  $\mathcal O$ of \fcts  is clear, we simply call  
a function belonging to
  $\mathcal O$ an \emph{\fct}.


We define weak transitions as follows:
$f\Arr{}g$ if 
there is a sequence of autonomous transitions $f\arr{\tau}f_1\arr{\tau}\dots\arr{\tau}f_n\arr\tau g$. $\Arr{\mu}$ is then simply $\Arr{}\arr\mu\Arr{}$.

We use state functions to express equations, such as $X = f(X)$. 
We thus look at conditions under which such an equation has a unique solution
(again, the generalisation to a system of equations is easy, {using $n$-ary functions}).


Thinking of functions as equation expressions,  to
 formulate our abstract theory about unique solution of 
equations, 
we have to  
define 
the divergences of finite and infinite unfoldings of state functions.
The $n$th unfolding of $f$ (for $n\geq 1$), $f^n$,   is the function obtained 
by  $n$ applications of  $f$.
An operator is \emph{well-behaved} if there is $n$ with
$f^n$ autonomous (the 
well-behaved operators 
 correspond, in   CCS, to   equations some finite unfolding of which 
yields  a guarded
expression). 
We also have to reason about the infinite unfolding of an equation $X = f(X)$.
For this, given a set 
  $\mathcal O$ of \fcts, we 
  consider the infinite terms obtained by infinite
  compositions of  \fcts\   in $\mathcal O$, that is, the set coinductively defined by the
  grammar: 
\iflong
\[  F :=  f\circ F  \hskip 1cm  \mbox{    for $f \in   \mathcal O$ }\]   
\else
$\quad  F :=  f\circ F  \hskip .5cm  \mbox{   where $f \in    \mathcal
  O$ (i.e., $f$ is a metavariable for the elements in $    \mathcal
  O$)}$.
\fi

\noindent 
{(Note that this syntax only deals with infinite
  compositions.  Since
$\mathcal O$   is closed under finite
 compositions, we do not need to handle such compositions in the
 above definition.) }
In particular, we write  $f^\infty$ 
for  the infinite term $f \circ f \circ f   \circ \ldots$.  

 We define the autonomous transitions for   
such infinite terms {using }the following rules:
 $$\displaystyle{   g \arr\mu   g' \over 
 g\circ F \arr\mu g' \circ F} \mbox{ $g$ autonomous}  \hskip 2cm 
 \displaystyle{    (g\circ  f) \circ F   \arr\mu F'   \over 
 g\circ ( f \circ F)  \arr\mu F'}  \mbox{
     $g $ not autonomous}$$
Intuitively  a term is `unfolded', with the second rule, until an autonomous
function is uncovered, and then  transitions are computed using the
first rule (we disallow unnecessary unfoldings; these would complicate our abstract
theorems, by adding 
 duplicate transitions, since 
the transitions
 of $g\circ f$ duplicate those of $g$
when $g$ is autonomous).
An infinite term has no transitions if none of its finite unfoldings
ever yields an autonomous function. This situation does not
arise for terms of the  
 form $f^\infty $ or $g\circ f^\infty $,
where $f$ is well-behaved, which are the terms we are  interested in.
Note that no
infinite term belongs to a set of operators.

\iflong

These rules are consistent for finite compositions of functions in
$\mathcal O$, in the sense that they allow one to infer the correct
transitions for finite compositions of functions.

\fi

The following lemmas show that infinite terms have the same 
 transitions $\arr\mu$ and weak transitions $\Arr\mu$ as their counterpart 
 finite unfoldings. To build a transition $f^\infty \arr\mu F$ 
 from a transition of an unfolding $f^n\arr\mu g$, 
 we need to reason up to (finite) unfoldings of $f$: thus we set $=_f$ to
be the symmetric reflexive transitive closure of the 
relation that relates
$g$ and $g'$ whenever $g=g'\circ f$. 

\begin{lemma}\label{l:abstract_trans}
Let $f$ be a well-behaved \fct and $g$ an \fct in some set $\Op$. 
Then $g\circ f^\infty\arr \mu F$ for some $F$  
if and only if there is $n$ such that $g\circ f^n\arr\mu h$ for some $h$; in which case, 
there is $h'=_f h$ such that $F= h'\circ f^\infty$.
\end{lemma}
\begin{proof}
\begin{enumerate}
\item[$(\Rightarrow)$] Assume $g\circ f^\infty\arr \mu F$. By a simple induction over the derivation of this 
transition, we get that there is $n$ such that $f^n$ is autonomous, $f^n\arr\mu h$, 
and $F= h\circ f^\infty$
\item[$(\Leftarrow)$] Assume $g\circ f^n\arr\mu h$. Since $f$ is well-behaved and composition respects 
autonomy, there is $m$ such that $g\circ f^m$ 
is autonomous; we take that $m$ to be the minimum among the exponents that make $g\circ f^m$ 
autonomous. If $n\leq m$, 
there is an autonomous transition $g\circ f^m\arr\mu h\circ (f^{m-n})$. 
If $m\leq n$, given that $g\circ f^m$ is autonomous, 
there is a transition $g\circ f^m\arr\mu h'$, 
where $h'\circ f^{n-m}=h$. Either way, $g\circ f^m\arr\mu h_0$, 
where $h_0=_f h$. 

Since $m$ is the smallest exponent that makes $g\circ f^m$ autonomous, we can derive 
$(g\circ f^m)\circ f^\infty \arr\mu g_0\circ f^\infty$. This concludes the proof.
\qedhere
\end{enumerate}
\end{proof}

{We now adapt Lemma~\ref{l:abstract_trans} to weak transitions.}
\begin{lemma}\label{l:abstract_weaktrans}
Let $f$ be a well-behaved \fct and $g$ an \fct in some set $\Op$. Then
$g\circ f^\infty\Arr \mu F$ for some $F$   
if and only if there is $n$ such that $g\circ f^n\Arr \mu h$ for some $h$; in which case, 
there is $h'=_f h$ such that $F= h'\circ f^\infty$.
\end{lemma}
\begin{proof}
For each implication, proceed by induction over the length of $\Arr\mu$, 
and use Lemma~\ref{l:abstract_trans}.
\end{proof}

\begin{definition}[\Fcts and divergences]\label{d:abs:div}
  Let $f,f',f_i$ be \fcts in a set $\mathcal{O}$ of \fcts, and
  consider the LTS induced by the autonomous transitions of
  \fcts in $\mathcal{O}$.
%
A sequence of transitions
$
f_1\arr{\mu_1}f_2 \arr{\mu_2}f_3\dots
$
is a \emph{divergence} if   
  for some $n\geq 1$ 
we have  $\mu_i = \tau $ whenever $i\geq n$.
We also say that $f_1$ \emph{diverges}.
We apply these notations and terminology also to infinite terms, 
 as expected. 
%
\end{definition}

In the remainder of the section 
we fix a set $\mathcal{O}$ of \fcts and we only
consider autonomous transitions on $\mathcal{O}$.
We  now state the ``abstract version'' of
Theorem~\ref{thm:roscoe_ccs}. 

\begin{theorem}[Unique solution, abstract formulation]
\label{thm:roscoe_abstract}
  Let $f$ be \hypabs \fct on $\mathcal O$
\iflong  on some LTS \Tlts. \fi
  If $f^\infty$ does not diverge, 
  then the equation $X=f(X)$ either has no solution or has a unique solution for $\bsim$.
\end{theorem}
\iflong
\begin{proof}
  The proof is essentially the same as that of Theorem \ref{thm:roscoe_ccs},
  replacing equations expressions with operators, from a fixed set of operators 
  $\mathcal{O}$ (one still has to consider reducts from an operator), and 
  replacing instantiation of equation expression transitions with instantiation 
  of autonomous transitions. A guarded \ee becomes an autonomous operator.
\end{proof}
\fi

The equation in the statement of the theorem might have no solution at
all. 
For
example, consider
 the LTS $(\N, \{a\}, \rightarrow)$ where for each $n$ we have
$ n+1\arr a n$. The function $f$ with $f(n)= n+1$ is an
\fct of the set $\mathcal{O}=\{f^n\}_{n\in\N}$ (with
$f^0 = \mathtt{Id}$, the identity function). The function $f$ is autonomous because, for all $n$,
the only transition of $f(n)$ is $f(n)\arr a n$ (this transition is
autonomous
because $f\arr a \mathtt{Id}$). A fixpoint of $f$ would be an
element $x $ with  $x\arr a x$, and there is no such $x$ in the LTS.


Theorem~\ref{thm:roscoe_abstract} 
can be refined along the lines of
Theorem~\ref{thm:roscoe2_ccs}. For this, we have to relate the
divergences of any $f^n$ (for $n\geq 1$)
to divergences of $f^\infty$, in order to distinguish between innocuous
and non-innocuous divergences.
%


\begin{lemma}
\label{l:inj}
  Consider an autonomous \fct $f$ on $\mathcal{O}$
  and a divergence of $f^n$ 
\[ f^n  \arr{\mu_1} f_1   \arr{\mu_2}  \ldots \arr{\mu_i} f_i 
\arr{\tau} f_{i+1}
\arr{\tau}  \ldots
 \] 
This yields a divergence of $f^\infty$: 
$
  f^\infty  \arr{\mu_1}g_1\circ f^\infty   \arr{\mu_2}  \ldots \arr{\mu_i} g_i\circ f^\infty 
\arr{\tau} g_{i+1} \circ f^\infty
\arr{\tau} \ldots
$

\noindent such that for all $i\geq 1$, $g_i$ is an operator and $g_i=_f f_i$.
\ifnot
\finish{below, the old version. Remove if the new one is always better; maybe see if we should rewrite}

{\alert  $\big((g_i,\mu_i,g_{i+1})\big)_{i\geq 0}\in \Div (f^n)$ ($g_0=f^n$).
  There exists a divergence 
  $\big((f_i\circ \fsol f,\mu_i, f_{i+1}\circ\fsol f)\big)_{i\geq 0}\in
  \Div(f^\infty)$
  such that $g_i=f_i\circ f^{n-i}$ for $i\leq n$, and $f_i=g_i$ for
  $i>n$.
  
This defines 
 the \emph{projection of divergences, $\phi_f$}, accordingly:
$$
\begin{array}{rlcl}
\phi_f:&\cup_{n\in \mathbb{N}^*}\Div (f^n) &\rightarrow
  &\Div(f^\infty)\\ 
       & \big((g_i,\mu_i,g_{i+1})\big)_{i\geq 0}\in \Div (f^n)
                                           &\mapsto&
                                                     \big((f_i\circ \fsol f,\mu_i, f_{i+1}\circ \fsol f)\big)_{i\geq 0}
\end{array}
$$}\fi

\end{lemma}
\iflong
\begin{proof} The proof is  a simple induction on the sequence of transitions defining 
the divergence. At each step, we unfold $f^\infty$ as many 
times as needed. 
\ifnot
\finish{idem}
{\alert
Consider a divergence of the \fct $f^n$: 
$f^n\arr{\mu_1}g_1\arr{\mu_2}g_2\dots$. Then, as $f$ is autonomous,
 there is $f_1$ such that $f\arr{\mu_1}f_1$ and $g_1= f_1\circ f^{n-1}$. 
Likewise, if $i\leq n$, there is 
$f_i$ such that $g_i=f_i\circ f^{n-i}$ and $f_i\circ f\arr{\mu _{i+1}}f_{i+1}$.
Therefore $f_i\circ\fsol f\arr{\mu _{i+1}}f_{i+1}\circ \fsol f$ for $i\leq n$.
Since also $f_n=g_n$, we have that $g_i\circ\fsol f\arr{\mu _{i+1}}g_{i+1}\circ \fsol f$.

  The sequence of the $f_i$'s gives us a matching divergence of $f^\infty$.}
  \fi
\end{proof}

Given a divergence $\Delta$ of $f^n$, we write $\Delta^\infty$ to
indicate  the
divergence of $f^\infty$ obtained from $\Delta$ 
as in 
Lemma~\ref{l:inj}. 
We call a divergence of  $f^\infty$ \emph{innocuous} when it can be
described in this way, that is, as a divergence 
$\Delta^\infty$  obtained from  a divergence $\Delta$ of $f^n$, for 
 some  $n$.

\begin{theorem}[Unique solution with innocuous divergences, abstract formulation]
\label{thm:roscoe2_abstract}
Let  $f\in \mathcal{O}$ be \hypabs \fct. If 
all divergences of $f^\infty$ are innocuous, then 
the equation $X= f(X)$ either has no solution or 
  has a unique solution for $\bsim$. 
\end{theorem}

\begin{proof}
Just as in CCS,   where 
 the proof of Theorem~\ref{thm:roscoe_ccs} 
 has to be modified for Theorem~\ref{thm:roscoe2_ccs},
here the proof of Theorem~\ref{thm:roscoe_abstract} is to be modified.
The modification is essentially the same as in CCS, 
again substituting equation expressions for 
operators. 
\end{proof}

\subsection{Reasoning with other Behavioural Relations}
\label{ss:equivalences}
\subsubsection*{Trace-based Equivalences.}

We can adapt the results of the previous section about
bisimilarity to  other settings, including both
preorders and non-coinductive relations.
As an example, we consider trace-based relations. 
%


We call a finite sequence of actions $s = \mu_1, \ldots, \mu_n$,
where each $\mu_i$ is a visible action, 
 a \emph{trace}. 
Accordingly, an infinite sequence of such actions $s=(\mu_i)_{i\in\N}$ is called an 
\emph{infinite trace}.
Given $s=\mu_1,\dots,\mu_n$ a trace, we write $x \Arr{s} $
if $x \Arr{\mu_1} x_1 \Arr{\mu_2}x_2 \ldots x_{n-1} \Arr{\mu_n}x_n$, for
some  states $x_1, \ldots, x_n$. Likewise, given $s=(\mu_i)_{i\in\N}$ 
an infinite trace, we write $x\Arr s$ if there are states $(x_i)_{i\in\N}$ such that 
$x=x_0$  
and for all $i\in\N$, $x_i\Arr{\mu_i} x_{i+1}$.

\begin{definition}[Trace-based relations]
\label{d:trace}
 Two  states $x,y$ are in the \emph{trace inclusion}, written 
$x \trincl y $, if $x \Arr{s} $ implies $y \Arr{s} $, for each trace
$s$. They are 
\emph{trace equivalent}, 
written 
$x \trequiv y $, if both 
$x \trincl y $ and 
$y \trincl x $ hold.

 Two  states $x,y$ are in the \emph{infinite trace inclusion}, written 
$x \itrincl y $, if $x \Arr{s} $ implies $y \Arr{s} $, for each finite or infinite trace
$s$. They are 
\emph{infinite trace equivalent}, 
written 
$x \itrequiv y $, if both 
$x \itrincl y $ and 
$y \itrincl x $ hold.
\end{definition}

The definitions from Section~\ref{s:lts:ops}
are the same as for $\bsim$; however we now consider sets of
\trincl-operators,
i.e., we are interested in \fcts that 
respect $\trincl$. Indeed the preorder $\trincl$ is used 
in the proof of unique solution for $\trequiv$, hence \fcts 
must respect $\trincl$ rather than $\trequiv$.
%

The notion of trace is extended to \fcts like we do for weak
transitions: we write $f\Arr{s}$ if 
there is a sequence of autonomous transitions $f\Arr{\mu_1}f_1\Arr{\mu_2}\dots\Arr{\mu_n}f_n$ 
($s=\mu_1,\dots,\mu_n$), and likewise for infinite traces. 

\begin{lemma}\label{l:trace}
Given a well-behaved \relfct \trincl $f$, a \relfct\trincl $g$,
 and a (finite) trace $s$, we have $g\circ f^\infty\Arr s$ 
if and only if there is $n$ such that $g\circ f^n\Arr s$.
\end{lemma}
\begin{proof}
We proceed by induction over the length of $s$, and apply Lemma~\ref{l:abstract_weaktrans} 
in each case.
\end{proof}

All theorems obtained for $\wb$  can be adapted to 
$\trequiv$, with similar proofs. As an example, Theorem~\ref{thm:roscoe2_abstract} becomes: 

\begin{theorem}\label{thm:roscoe2_trace}
  Let $f\in \mathcal{O}$ be \hypabs \relfct \trincl. If 
  all divergences of $f^\infty$ are innocuous,
  then the equation
  $X= f(X)$ either has no solution or has a unique solution for $\trequiv$.
\end{theorem}

\begin{proof}
  We proceed by showing that $x\trincl f(x)$ implies
  $x\trincl f^\infty$, and then that $f(x)\trincl x$ implies
  $f^\infty \trincl x$. This indeed gives that $x\trequiv f(x)$
  implies $x\trequiv f^\infty$; hence the equation has at most one
  solution ($f^\infty$ does not belong to the LTS, hence it is not a
  solution).
\begin{enumerate}
\item \label{proof1_trace}
$f(x)\trincl x$ implies $f^\infty \trincl x$. For this part, the absence of divergence  hypothesis is not needed.

Let $s$ be a trace, and assume $f^\infty \Arr s$. By Lemma~\ref{l:trace}, there is 
$n$ such that $f^n\Arr s$. Hence, $f^n(x)\Arr s $. Since $f(x)\trincl x$ and $f$ is an 
\relfct \trincl, we have that $f^n(x)\trincl f^{n-1}(x)\trincl\dots \trincl x$. 
Hence, $x\Arr s$, and 
$f^\infty \trincl x$.
\item \label{proof2_trace}
$x\trincl f(x)$ implies $x\trincl f^\infty$. This part of the proof is very similar 
to the proofs of Theorems~\ref{thm:roscoe2_ccs} and~\ref{thm:roscoe2_abstract}.

Assume $x\trincl f(x)$, and $x\Arr s$. We want to show that there exists some $n$ 
such that $f^n \Arr s$, then apply Lemma~\ref{l:trace}; this would show $f^\infty \Arr s$.
To that end, we build a strictly increasing sequence 
of (autonomous) transitions $f^{n\times p}\Arr {s_n} g_n$,
such that $g_n(x)\Arr {s'n}$ and such that $s=s_n s'_n$ for all $n$ ($s_n$ and $s'_n$ are 
traces whose concatenation is $s$). Here $p$ is such that $f^p$ is 
autonomous (we always unfold $p$ times at once).
Strictly increasing means that the transition $f^{(n+1)\times p}\Arr{s_{n+1}} g_{n+1}$
can be decomposed as $f^{(n+1)\times p}\Arr{s_n}g_n\circ f^p\Arr {s''_n} g_{n+1}$ for some 
$s''_n$ ($s_{n+1}=s_n s''_n$).
This construction will have to stop, otherwise we 
build a non-innocuous divergence.

\subsubsection*{Construction of the sequence.}
Assume $p$ is such that $f^p$ is autonomous; such a $p$ exists, since $f$ is 
well-behaved. 

We initialise with the empty trace from $f^0=\identity$. Indeed, we have 
$\identity(x)\Arr{} \identity(x) \Arr{s} $, and $s$ is indeed the concatenation 
of itself with the empty trace.

Then, at step $n$, suppose we have, for instance,
$f^{n\times p} \Arr{s_n}g_n$, and $g_n(x)\Arr{s'_{n}}$.

%
\begin{itemize}
\item If $s'_n$ is the empty trace, we stop. We
  have in this case $f^n\Arr s $, which concludes.
\item Otherwise, we have 
  $f^{(n+1)\times p}\Arr{s_n} g_n\circ f$ (by autonomy). 
  The \fcts are \relfcts \trincl, and by hypothesis $x\trincl f(x)$, 
  therefore $g_n(x)\trincl g_n\circ f^p(x)$. 
  From there, $g_n\circ f^p(x)\Arr{s'_n}$. 
  
  $s'_n$ is not the empty sequence, so there is  $y_1,\dots,y_n$ for 
  $n\geq 1$ such that $g_n\circ f^p(x)\arr{\mu_1} y_1\arr{\mu_2}\dots\arr{\mu_n}y_n$, 
  where $s'_n=\mu_1\dots\mu_k$.
  Take $\mu_1\dots\mu_i$ the longest prefix of $s'_n$ such that there is \fcts
   $h_1,\dots,h_i$ such that $g_n\circ f^p\arr{\mu_1} h_1\arr{\mu_2}\dots\arr{\mu_i}h_i$.
  It cannot be empty: indeed, $f^p$ is autonomous and 
  composition respects autonomy, hence so is $g_n\circ f^p$, so $g_n\circ f^p(x)\arr{\mu_1} y_1$ 
  is always autonomous. 
  We take $g_{n+1}$ to be $h_i$, 
  $s'_{n+1}$ to be the trace corresponding to $\mu_{i+1}\dots\mu_k$ (i.e., removing all $\tau$s) 
  and $s_{n+1}$ to be $s_n s''_n$, where $s''_n$ is the trace corresponding to $\mu_1\dots\mu_i$ 
  (again, removing all $\tau$s). 
  
  Then, we indeed have $f^{n\times p}\circ f^p=f^{(n+1)\times p}\Arr{s_{n+1}} g_{n+1}$, 
  and $g_{n+1}(x)\Arr{s'_{n+1}}$, where $s=s_{n+1} s'_{n+1}$.
\end{itemize}

Suppose now that the construction given above never stops. To make  
the argument clearer, we reason up to $=_f$ (we identify all \fcts that 
are $=_f$). 
We know that
$g_n\circ{f^p} \Arr{s''_n} g_{n+1}$, therefore, by applying 
Lemma~\ref{l:abstract_weaktrans} as many times as needed (just as in Lemma~\ref{l:trace}), 
we get that $g_n\circ f^\infty \Arr{s''_n} g_{n+1}\circ f^\infty$.

 This gives an infinite
sequence of transitions starting from $ f^\infty$:
$f^\infty \Arr{s''_0}g_1\circ f^\infty\Arr{s''_1}\dots$. We observe that
in the latter sequence, every step involves at least one transition (the 
$s''_i$ are supposed non-empty, otherwise we would have stopped).
Moreover, the sequence $s''_1 s''_2\dots$, when removing all $\tau$s, 
is a prefix of $s$. 
Therefore there is a finite number of visible actions occurring in this infinite
sequence.  Therefore $f^\infty$ is divergent. 

Furthermore, this 
divergence is not innocuous: at every step, we know that the 
transition $g_{n+1}(x)\Arr{s''_{n+1}} $ is not autonomous, 
otherwise $s''_{n+1}$ would be part of $s_n$. Hence, 
there cannot be such a divergence $f^m\Arr{s''_0}\Arr{s''_1}\dots$, 
as this would imply $g_{n+1}(x)\Arr{s''_{n+1}} $ is autonomous for 
$n$ such that $n\times p \geq m$.
\qedhere
\end{enumerate}
\end{proof}

In contrast, the  theorems fail for
\emph{infinitary trace equivalence}, \itrequiv (whereby two processes 
are equated if they have the same traces,  
including the infinite ones), for
the same reason why the `unique solution of contraction' technique
fails in this case~\cite{popl15}.  As a counterexample, 
we consider 
equation $X =  a+\inpC a . \X$, whose syntactic solution has no
 divergences. 
The process $P\Defi \sum_{n>0}\inpC {a^n} $ is a solution, 
yet it is not \itrequiv-equivalent 
 to the syntactic solution of the
equation, {because the syntactic solution has an infinite trace
  involving $a$ transitions}.
{This phenomenon is due to the presence of 
  infinitary observables.}


\subsubsection*{Preorders.}
\newcommand{\twosim}{\asymp}
\newcommand{\ws}{\leq_{s}}

We show how the theory 
for  equivalences can be transported  onto \emph{preorders}. 
This means moving to \emph{systems of  pre-equations},
$\{  X_i \leq E_i\}_{i\in I}$.
With preorders,  
our theorems have a different shape: we do not use pre-equations to
reason about unique solution -- we expect interesting pre-equations to
have many solutions, some of which may be  incomparable with each
other.
We rather derive  theorems to prove that, in a given preorder, 
 any solution of a  pre-equation is below its
syntactic solution.

\iflong
{In the LTS we consider, an equation does not always have a solution.
We thus 
extend the original LTS with the transitions corresponding to the autonomous transitions of the syntactic solution $f^\infty$.
%
}


\fi

We write $\trincl$ for trace inclusion, $\itrincl$
for infinitary trace inclusion, 
and $\ws$ for weak simulation. 
These preorders 
are standard from the literature~\cite{spectrum}.

In the abstract setting, the body of the pre-equations are functions. 
Then the theorems give us conditions under which, given a pre-equation 
 $X \leq f(X) $ and a  behavioural  preorder $\preceq$,
a  solution $r$, i.e., a state 
 for which $r \preceq f(r) $ holds, 
is below  the syntactic solution $f^\infty$. 
We present 
the counterpart of Theorem~\ref{thm:roscoe2_abstract}; other theorems are
transported in a similar manner\iflong, both the statements and the  proofs.  
\else
.
\fi

\begin{theorem}
\label{t:preorder}
  Let $f\in \mathcal{O}$ be \hypabs \fct. 
  If  
  all divergences of $f^\infty$ are innocuous,
  then, given  a preorder 
${\leq}\in {\{\trincl, \itrincl,\ws \}}$, 
whenever  $x\leq f(x)$ we also have $x\leq f^\infty$, for any  state
$x$.
\end{theorem}

\begin{proof}
\begin{enumerate}
\item For $\trincl$, the proof is given as the item~\ref{proof2_trace} of the proof 
of Theorem~\ref{thm:roscoe2_trace}.
\item For $\itrincl$, the proof is very similar to the previous proof: simply consider 
infinite traces, and rather that disjunct over whether or not the construction stops (the 
construction cannot stop), disjunct over whether the full trace is captured, or 
only $\tau$ actions occur from a moment in the construction.
\item For $\ws$ The proof is strictly included in the proofs of 
Theorems \ref{thm:roscoe2_ccs} or \ref{thm:roscoe2_abstract}.
\qedhere
\end{enumerate}
\end{proof}

Theorem~\ref{t:preorder} intuitively says that 
 the  syntactic solution 
of a  pre-equation
is \emph{maximal} 
among all solutions.  Note that, 
in contrast with equations, Theorem~\ref{t:preorder} and 
the theory of pre-equations also work for 
\emph{infinitary trace inclusion}.

The opposite direction for pre-equations, namely 
$\{  X_i \geq E_i\}_{i\in I}$ is less interesting; 
 it means 
 that the syntactic solution is \emph{minimal} among the solutions.
This property is usually easy to obtain for a behavioural
preorder, and we do not need hypotheses such as autonomy,
well-behavedness, 
or 
non divergence, as stated in the following proposition:


\begin{proposition}
\label{prop:min}
Let $f\in\mathcal{O}$ be an operator, $\leq \in \{\trincl, \ws\}$, and $x$ 
such that $f(x)\leq x$. Then, $f^\infty \leq x$.
\end{proposition}
\fi

\begin{proof}
\begin{enumerate}
\item For $\trincl$, the proof is given as item~\ref{proof1_trace} of the proof of 
Theorem~\ref{thm:roscoe2_trace}.
\item For $\ws$, the proof is very similar to the above proof for $\trincl$.
  \qedhere
\end{enumerate}
\end{proof}


However, Proposition~\ref{prop:min} may fail for preorders with
 infinitary observables, such as infinitary trace inclusion.
\iflong
This is why the theory of equations fails for infinitary trace
equivalence. 
Indeed, if $P$ is a solution for the equation $E\itrincl \X$ (whether or not
 $E$ is  divergence-free), 
then 
$ \ssol E \itrincl P$ does not necessarily hold.
The equation $X = a+a.X$, which is discussed after
Theorem~\ref{thm:roscoe2_trace}, provides a counter-example to
illustrate this observation.


\begin{remark}
  Suppose that $x$ and $y$ are such that $x\leq f(x)$ and $y \leq
  x$. Then we do not necessarily have $y\leq f(y)$.
Take for instance, in CCS, $P= a .b$, $Q = a|b$ and $E = a.X+b.X$. Then $Q\leq_s E[Q]$, and 
$P\leq_s Q$, however $P\not\leq_s E[P]$.
\end{remark}
\subsection{Rule formats}

A way to instantiate the results in Sections~\ref{s:lts:ops} and~\ref{ss:equivalences}
 is to consider
\emph{rule formats}~\cite{sos:handbook,DBLP:journals/tcs/MousaviRG07}. 
These provide a specification for the form of the SOS rules  used to
describe the constructs  of a  language.  
To fit a rule format into the abstract formulation of the theory from Section~\ref{s:lts:ops},
we view the constructs of a language  as functions on  
the states of the LTS (the terms or processes of the language).

One of the most common formats is
GSOS~\cite{DBLP:journals/jacm/BloomIM95,gsos,gsos_weak,DBLP:conf/lics/FokkinkG16}. 
\iflong
The format guarantees that the constructs of the 
 language preserve autonomy. This allows to show that the set of unary contexts 
 (meaning, contexts that can be filled with one term) 
 in a GSOS language constitute a set of operators (Definition~\ref{d:opsets} 
 (when seen as functions from terms to terms), 
 and allows us to apply Theorem~\ref{thm:roscoe2_abstract} to any GSOS 
 language.
\fi

We now use $x,y$ for \emph{state variables}, and $t,u$ for terms that 
are part of the considered GSOS language (i.e., for elements of the set of states). 
We use $c,d$ for contexts, 
seen as functions from the set of states to itself. We write $\op$ for a construct 
of the language (a function symbol).



For a complete overview of rule formats and for the relevant definitions, 
we refer the reader to \cite{gsos}. We recall the format of GSOS 
rules below, slightly modified in order to use contexts: 
\begin{mathpar}
 \inferrule{\{x_i\arr{\mu_{i,j}}y_{i,j}\st i\in I,~1\leq j\leq m_i\}\cup 
 			\{x_j\notarr{\mu'_{j,k}}\st j\in J,~1\leq k\leq n_j\}}
 		{\op(\til x)\arr\mu c(\til x,\til y)}
\end{mathpar}
$I,J$ are fixed subsets of $[1,n]$, where $n$ is the length of
$\til x$. For $i\in I$ (resp.\ $j\in J$), $m_i$ (resp.\ $n_j$) is a
fixed integer. 
$\til y$ is the list consisting of all
$y_{i,j}$ for $i\in I$ and $1\leq j \leq m_j$, as well as all
$y_{j,k}$ for $j\in J$ and $1\leq k \leq n_j$. $c$ is a 
context belonging to the language.

\begin{lemma}\label{l:gsos}
The set of unary contexts $\mathcal{O}$ of a language defined within
the GSOS format 
 is such that composition in $\mathcal{O}$ preserves autonomy on $\mathcal{O}$.
\end{lemma}

\begin{proof}
Let $\mathcal{L}$ be such a language. We reason by induction over the
contexts (which are defined inductively using the constructs of
$\mathcal L$).
\begin{itemize}
\item The empty context (i.e., the `hole', corresponding to the identity function $\identity$) preserves 
autonomy: if $c\in \mathcal{O}$ is autonomous, then so is $\identity \circ c= c$.
\item Consider a $n$-ary construct of the language (a function symbol), $\op$, 
and a context $c=\op(c_1,\dots,c_n)$, 
where the $c_i$'s are unary contexts. 

We have to show that so is $c\circ c'$, for  $c'$ autonomous.
Consider for this a transition $c\circ c' (t)\arr\mu u$
($t,u\in\mathcal{L}$), we prove that this transition is autonomous.
The last rule of the derivation tree of this transition is an instance
of a GSOS rule as seen above:
\begin{mathpar}
 \inferrule{\{x_i\arr{\mu_{i,j}}y_{i,j}\st i\in I,~1\leq j\leq m_i\}\cup 
 			\{x_j\notarr{\mu'_{j,k}}\st j\in J,~1\leq k\leq n_j\}}
 		{\op(\til x)\arr\mu c_0(\til x,\til y)}
\end{mathpar}
($I, J$ and the $m_i$s and $n_j$s being fixed, as well as $c_0$).

By definition $c\circ c'(t)=\op(c_1\circ c'(t),\dots,c_n\circ c'(t))$. 
Therefore the $x_i$ must be instantiated by the terms $c_i\circ c'(t)$. 
Hence there are $u_{i,j}\in\LL$ such that $c_i\circ c'(t)\arr{\mu _{i,j}} u_{i,j}$ 
for $i\in I,~1\leq j\leq m_i$. As well, $c_j\circ c' (t)\notarr {\mu _{j,k}}$ 
for $j\in J,~1\leq k\leq n_j$. Lastly, writing $u=c_0(\til c\circ c'(t),\til u)$ and 
$c\circ c'(t)\arr\mu c_0(\til c \circ c'(t),\til u)$.

By induction hypothesis, each of the $c_i$'s preserves 
autonomy under composition on $\Op$, hence $c_i\circ c'$ is autonomous for all $i$.
Hence there are autonomous  transitions $c_i\circ c'\arr{\mu _{i,j}} d_{i,j}$, 
where $d_{i,j}$ are contexts such that $u_{i,j}=d_{i,j}(t)$. 

Furthermore, for any $t'\in\LL$, $c_j\circ c' (t')\notarr {\mu _{j,k}}$: if there 
was such a $t'$, by autonomy of $c_j\circ c'$, we would have $c_j\circ c'(t)\arr{\mu _{j,k}}$. 
 
Therefore, for any $t'\in\LL$, the premises of the above GSOS rule hold; 
hence we can deduce $c\circ c' (t')\arr\mu c_0(\til c \circ c'(t'),\til d(t))$. 
Thus there is an autonomous transition $c\circ c'\arr\mu c_0(\til c \circ c',\til d)$, 
and $u= (c_0(\til c \circ c',\til d))(t)$. The transition $c\circ c' (t)\arr\mu u$ is 
indeed autonomous. This concludes the proof.
\qedhere
\end{itemize}
\end{proof}

\begin{proposition}\label{prop:gsos}
Consider the
 set of unary contexts 
of a language defined within the GSOS formats, and   suppose that
bisimilarity is preserved by these contexts. This set constitutes a 
set of operators.
\ifnot 
\textbf{old version , more general:}
\textbf{version alternative, plus generale:}\danielmain{ok pour cette
  seconde version, en indiquant que la premiere version (qui est celle
  qui nous interesse, I suppose) s'obtient comme corollaire}
A set of contexts $\mathcal{O}$ 
of a language defined within the GSOS formats that is closed under 
composition, contains the identity and 
such that bisimilarity is preserved by the contexts of $\mathcal{O}$, constitute a 
set of operators.
\fi
\end{proposition}

\begin{proof}
This is a direct consequence of Lemma~\ref{l:gsos}: 
for any language, 
the empty context acts as the identity function over the set of terms, 
and by construction contexts can be composed; by hypothesis, we consider contexts 
that respect bisimilarity. All is left to check is clause~\ref{d:operators:autonomy:comp} of Definition~\ref{d:opsets}, 
i.e., that composition preserves autonomy over $\Op$: this is give by 
Lemma~\ref{prop:gsos}.
\end{proof}

Some GSOS rule formats guarantee congruence for weak
bisimilarity~\cite{gsos_weak,DBLP:conf/lics/FokkinkG16}, which
gives {clause~\ref{d:operators:congruence} of Definition~\ref{d:opsets}}. 
{This condition
  could actually be relaxed,
  by imposing congruence only for certain forms of contexts, according
  to how variables occur in the equations. For instance, 
  if certain syntactical constructs are
  not used, we do not need congruence to hold for these constructs. 
%
This can be useful, for instance, to allow non-guarded sums in CCS: 
in sums, we require variable occurrences to be under prefixes (and the
sum itself need not between prefixed processes).
}

\begin{remark}
Proposition~\ref{prop:gsos} makes it possible to apply 
Theorems~\ref{thm:roscoe_abstract} and~\ref{thm:roscoe2_abstract} 
to any GSOS language that additionally preserves bisimilarity.
However the formulation of the `syntactic solution of an equation' and of `innocuous
divergence'  in Section~\ref{s:lts:ops} is not as easy to grasp as in CCS. 

An alternative would be 
 to assume that the language has constants for recursively defined
 processes, which would allow us to remain closer to the setting of
 CCS, as exposed in Section~\ref{s:mt}.
\end{remark}


  

\ifnot
\begin{theorem}\label{t:gsos}
Consider a language whose constructs have SOS rules in the  GSOS
format and preserve $\wb$,  
 and an equation $X =E $ for the language, where $E$ is a function
 corresponding to BLABLA
 If $E$ is autonomous 
(over the set of functions corresponding to the 
 contexts of the language), 
 and if, in the language extended with constants, 
the syntactic solution of the equation
 only has innocuous divergences,
then either the equation has no
solution or it has a unique solution for $\bsim$.
\end{theorem}
\fi


We now state the abstract version of Theorem~\ref{thm:roscoe2_ccs}. For this, we rely
on  the definition of divergence and innocuous divergence 
from Section~\ref{s:lts:ops}, and use contexts of a GSOS language to
stand for operators.

\begin{theorem}
Consider a language whose constructs have SOS rules in the  GSOS
format and preserve $\wb$. Consider an equation of the form 
 $X =E$, where $E$ is a function corresponding to a context of the
 language (seen as an equation expression), that is autonomous and has  only innocuous divergences, 

Then either the equation has no
solution or it has a unique solution for $\bsim$.
\end{theorem}
\begin{proof}
This is a direct consequence of Theorem~\ref{thm:roscoe2_abstract} and 
Proposition~\ref{prop:gsos}. Consider the context $C$ such 
that $C[X]$ and $E$ are syntactically equal. Then $C\in\Op$, 
where $\Op$ is the set of unary contexts of the language; by Proposition~\ref{prop:gsos}, 
it is a set of operators. Using Theorem~\ref{thm:roscoe2_abstract} we deduce that the 
equation has a unique solution.
\end{proof}

Checking 
the autonomy property  
 is often
 straightforward; for instance, 
it holds if, in the body of an equation, all variables are underneath
 an axiom construct, that is, a construct that (like prefix in CCS) is defined by
 means of SOS rules in which the set of premises  is empty. 
{Definitions of guardedness for SOS specifications from~\cite{guardedness} can also 
 be applied here.}

%

\subsubsection*{tyft/tyxt formats.}\label{ss:tyft}
We briefly discuss other widely studied formats, the tyft/tyxt formats~\cite{tyft}. 
In those formats, lookaheads are possible. 
Lookahead allows one to write rules that 
`look into the future' (a  transition is allowed if  certain sequences of
actions are possible); this breaks autonomy (condition~\ref{d:operators:autonomy:comp} of Definition~\ref{d:opsets}), 
hence Theorem~\ref{thm:roscoe2_abstract}
 does not hold in a tyft/tyxt language.
 We provide a counter-example,
in a language with three constructions: first, a prefix constructor,
like the prefix of CCS, and with the same rule (it is indeed a tyft
rule). We assume that we can use the prefix construction with at least two distinct labels, $a$ and $b$. Second, a constant $\mathbf{0}$, that has no rule (as in CCS). Third, a constructor $\forwd$ with the following tyft rule:
$$\displaystyle{ x\arr{\mu'} y ~~~~y\arr\mu z
\over {\forwd (x)\arr\mu z }} $$
Then, the equation $X= a.\forwd (X)$ is autonomous or (strongly) guarded, but does not enjoy unique solution for $\bsim$: both $a.\mathbf{0}$ and $a.b.\mathbf{0}$ are solution for $\bsim$.

\section{Name Passing: the \pc}
\label{s:pi}

\subsection{Unique solution in the asynchronous \pc}\label{s:usol:api}

In this section, we port our results onto the asynchronous \pc, \Apc~\cite{SW01a}
(addressing the full $\pi$-calculus would also be possible, but somewhat more
involved).
The theory of \Apc\ is simpler than that of the synchronous
$\pi$-calculus. Notably, bisimulation does not require closure under
name instantiation.

\subsubsection*{Equation expressions,  processes and systems of equations.} 
For the $\pi$-calculus, we inherit the notations from CCS. 
 $a,b, \ldots, x,y, \ldots$  
range  over the infinite set of names, and $X, Y, \ldots$ range over
equation variables.


%

{In the $\pi$-calculus, the solutions of equations are \emph{process
abstractions}, i.e., processes which are made parametric over their
free names. This allows us to work with closed agents, which makes the treatment of
equations easier (in particular w.r.t.\ name capture).
  }

In order to define the syntactical entities involved in the reasoning
about the \pc, it is useful to start by introducing equation expressions,
ranged over using $E, E', \ldots$,
as follows:
\begin{align*}
F & ~:=~ K \midd X \tag*{(abstraction identifiers)}
     \\[.3em]
E & ~:=~ \nil    \midd    \inp a \tilb . E    \midd    \out a \tilb 
   \midd     
  E_1 |  E_2   \midd    \res a E        
 \midd ! \inp a \tilb . E \midd \app F \tila \tag*{(equation expressions)}
\end{align*}

Inputs of the form $\inp c \tila . E$ and  $! \inp c \tila . E$, 
and restrictions $\res a E$ are
binders for the names
$\tila$ and $a$, with scope $E$. 


Replication could be avoided in the syntax since it can be encoded with recursion. However
 replication is a useful construct for  examples and it is therefore convenient to
have it as a primitive in
the syntax because some of the conditions for our unique-solution theorems 
will then be easier to check (we only allow replicated inputs, which
are guarded).

The grammar for equation expressions includes the
application construct $\app F\tila$, used to instantiate the formal
parameters of the abstraction identifier $F$ with the actual
parameters $\tila$. Indeed, since the operational semantics of the
$\pi$-calculus makes use of name substitution, 
the body of recursive definitions and of equations is
parametrised over a set of names (as will be apparent below).

\emph{Processes}, ranged over using $P, Q, \ldots$, are defined as
equation expressions in which no occurrence of equation variables is
allowed (therefore the grammar for processes, $P$, is like the grammar
for $E$ above, except that $K$ replaces $F$).

 \emph{Process abstractions} are of the form $\abs\tila P$, where
 $\tila$ is a non-empty tuple of distinct names. In $\abs\tila P$,
 names $\tila$ are bound in $P$, giving rise to the usual notion of
 bound and free names. A process abstraction is \emph{closed} if it
 has no free name.

 \emph{Agents.} 
Processes, process abstractions, and  constant identifiers
form the set of \emph{agents}: 
\[
A \quad := \quad P\midd \abs \tila P \midd K 
\hspace{2cm} \tag*{\mbox{(agents)}}
\]

An agent is \emph{closed} if it does not have free 
names. 
\emph{Recursive definitions.}  $K$ is used to stand for recursively
defined abstractions. Each $K$ should have a defining equation of
the form $K \Defi \abs\tila P$, with the additional constraint that $\abs\tila P$ is
closed (i.e., it has no free name).
  
  

\emph{Systems of equations.} In order to write equations,
we need to work with \emph{equation abstractions}, written
$\abs\tila E$, where $\tila$ is a non-empty tuple of distinct
names. 
%
%
Again, in $\abs\tila E$, names $\tila$ are bound in $E$, giving rise
to the usual notion of bound and free names. We say that $\abs\tila E$
is \emph{closed} if it has no free name.

An equation is of the form $X = \abs\tila E$, where the body
$\abs\tila E$ of the equation is
closed, and can contain the equation variable $X$.
%
A solution of such an equation is given by a \emph{closed} process abstraction of
the form $\abs\tilb P$ (or, in the case where $\tila$ is empty, by a
closed process).  
Imposing this condition allows us to ignore difficulties related to
name capture (if, for instance,  we substitute (non closed) $P=\inp
ax.Q$ for $X$ in the equation expression $E=\res a\,X$, name $a$ gets
captured by the restriction).





Equation variables occurring in equation expressions can be
instantiated with process abstractions. When doing so, we instantiate
names in the process abstractions according to the usages of the
equation variable in the equation expression; this operation is done
implicitly, and does not correspond to a computation step. We
illustrate this on an example:
\begin{example}
  Consider the equation expression $E=\out c{d} | \app X
  \tila$, having $X$ as only equation variable.  
  Then
  $E[\abs{\tilb} P]$ stands for
  $\out cd | P\{\tila/\tilb\}$.
%
Note that the equation associated to equation expression $E$ is of the
form $X = \abs{c,d,\tila}E$.
\end{example}


We can also instantiate equation variables 
with equation abstractions: if $E_1$ is an equation expression and
$\abs\tila E_2$ is an equation abstraction, 
then $E_1[E_2]$ stands for
the equation expression obtained by instantiating the variable in
$E_1$ with $\abs\tila E_2$.

The definitions above for equations are extended to systems of
equations; in particular, the solution of a system of equations is
given by a tuple of process abstractions.

As in CCS, we can turn equations into recursive definitions; this
yields the \emph{syntactic solution} of the equations.






%

\ifnot
\finish{i wonder if replication should be a primitive. It can be encoded, but we use it
  also inside equations, and then the encoding becomes a bit delicate, in particular if we
say things like `if equations are guarded then...', if there are replications that have to
be encoded the statement becomes a bit obscure} 
\daniel{we apparently only need replicated inputs (to phrase
  example~\ref{ex:param:api}  in this subsection, and to  treat the example of lambda
  encoding); we could discuss only the encoding of replicated inputs,
  and notice that it gives rise to a guarded definition.}
\fi
%

\emph{Sorting.} Since the calculus is polyadic, we assume a
\emph{sorting system} \cite{Mil99} to avoid disagreements in the
arities of the tuples of names carried by a given name.  The sorting
is extended to agents in the expected manner.  Similarly, when
considering equation variables, it is intended that each variable has
a sort so that 
we know the number and the sorts of their parameters.  An
equation expression is also sorted, 
and, accordingly, composition of equation
expressions is supposed to be well-sorted.

We will not present the sorting system 
because  it is well-studied (cf.~\cite{Mil99,SW01a}), and adapting it to our
setting raises no particular difficulty. 
The reader should
take for granted that all agents described  obey  a sorting. 

\subsubsection*{Transitions.} 
As in CCS, transitions are of the form $P\arr\mu P'$, where the
grammar for actions is given by
$$
\mu\quad:=\quad \inp a\tilb\midd \res{\til{d}}\out a\tilb\midd \tau
\enspace.
$$

$\res{\emptyset}\out a\tilb$ is written simply $\out a\tilb$. Free and
bound names of an action are defined as usual.
%
%
Figure~\ref{f:lts:api} presents the transition rules for \Apc.
We maintain from CCS 
 the notations for transitions, such as 
$\Arr {}$, $\Arcap \mu$, and so on.

\begin{figure}[ht]
\begin{mathpar}
  \inferrule{~}{\inp a\tilb.P \arr{\inp a\tilb}P}
  \and
  \inferrule{~}{!\inp a\tilb.P \arr{\inp a\tilb}!\inp a\tilb.P | P}
  \and
  \inferrule{~}{\out a\tilb\arr{\out a\tilb}\nil}
  \and
  \inferrule{P\arr{\res{\til{d}}\out a\tilb}P'
  }{
    \res n P\arr{\res{(\{n\}\cup\til{d})}\out a\tilb} P'}
  n\in\tilb
  \and
  \inferrule{P\arr\mu P'
  }{
    \res n P\arr\mu \res n P'}
  d\notin\names{\mu}
  \and
  \inferrule{P\arr{\inp a \tilb}P'\and Q\arr{\res{\til{d}} \out a{\til{b'}}}Q'
  }{
    P | Q\arr\tau \res{\til{d}}(P'\sub{\til{b'}}{\tilb} | Q')
  }
  \and
  \inferrule{P\arr\mu P'}{P | Q\arr\mu P' | Q}\bnames\mu\cap\fnames
  Q=\emptyset
  \and 
  \inferrule{P\{\tilb/\tila\}\arr\mu P'}{\app K\tilb\arr\mu P'} \mbox{ if }K\Defi \abs\tila P
\end{mathpar}
  \caption{\Apc: Labelled Transition Semantics}
  \label{f:lts:api}
\end{figure}

In bisimulations or similar coinductive relations for the asynchronous
$\pi$-calculus, no name instantiation
is required in the input clause or elsewhere (provided
$\alpha$-convertible processes are identified); i.e., the
\emph{ground} versions of the relations are congruence relations
 \cite{SW01a}. 
Similarly, 
the extension of bisimilarity to agents only considers fresh names:
 $A  \wb A'$ if $\app{A}{\tila} \wb \app{A'}{\tila}$ where $\tila $ is  a tuple 
of fresh names (as usual, of the appropriate sort).


\begin{example}\label{ex:param:api} 
If $K$ is the syntactic solution of the 
 equation $X=(a)~!a(x).X\narg x$, then 
we have  $\app K b \arr{b(y)} \app  K y |  \app K b $.
\end{example}

Theorems~\ref{thm:roscoe_ccs} and~\ref{thm:roscoe2_ccs} for CCS can be
adapted to the asynchronous $\pi$-calculus. 
The definitions concerning transitions
and divergences are transported to \Apc 
as
expected. In the case of an abstraction, 
one first has to
instantiate the parameters with fresh names; thus $F$ has a divergence
if the process $\app F\tila$ has a divergence, where $\tila$ are fresh
names.

\begin{theorem}[Unique solution in  A$\pi$]\label{thm:roscoe_api}
A guarded system of equations 
whose  syntactic 
solutions do not contain divergences has  a unique 
solution for $\bsim$.\end{theorem}

\begin{proof}
Because we only consider closed abstractions, all names are
instantiated when performing a substitution of a variable by an
abstraction.  This allows us to use the same proof as for
Theorem~\ref{thm:roscoe2_ccs} in CCS.
\end{proof}

\begin{theorem}[Unique solution with innocuous divergences in  A$\pi$]\label{thm:roscoe2_api}
A guarded system of equations 
whose  syntactic 
solutions only have  innocuous divergences has  a unique 
solution for $\bsim$.
%
\end{theorem}

\begin{proof}
  The proof of Theorem~\ref{thm:roscoe_api} has to be modified exactly
  as in the CCS case, where the proof of Theorem~\ref{thm:roscoe_ccs} is
  modified to establish Theorem~\ref{thm:roscoe2_ccs}.
\end{proof}

As in CCS, the guardedness condition can be removed
if the rule for the unfolding of a constant
 produces a $\tau$-transition. 

We now state an analogue of Lemma~\ref{l:criterion:div} for \Apc, giving a  
sufficient condition to guarantee that a system of equations only has innocuous 
divergences.
{This condition is decidable, and sufficient for the 
application in Section~\ref{s:application_cbn}. We leave the study of finer 
conditions for future work.}

\begin{lemma}\label{l:div_api}
Consider a well-sorted system of equations $\til{X} =\til{E}$ in \Apc\
(in particular,  for each $i$, the sort of $X_i$ and of $E_i$ are the same).

Suppose that there is a sort of names such that names of that sort are never 
used in subject output position.
If for each $i$ there is
$n_i$ such that in $(E_i)^{n_i}$, each equation variable occurs underneath an 
(input) prefix of that sort, then the system has only innocuous
divergences.
\end{lemma}

In earlier sections about CCS, we pointed out the connection between techniques based on
unique solution of equations and 
enhancements of the bisimulation proof
method. The same connection is less immediate in name-passing calculi, where indeed 
there are noticeable  differences.
In particular,   `up to context' enhancements for the ground bisimilarity of the
$\pi$-calculus require closure under name instantiation, even when ground bisimilarity
is known to be  preserved by substitutions (it is an open problem whether the closure can
be lifted). Thus, when comparing two derivatives $\ct P$ and $\ct Q$,
in general it is not sufficient that $P$ and $Q$ alone are in the
candidate relation:  
 one is required to include 
  also  all their  closures under name
substitutions (or, if the terms in the holes are abstractions,
instantiation of their parameters with arbitrary tuples of names). 
In contrast, the two unique solution theorems above are `purely ground': 
$F =  (\tilx)P$ is solution of  an equation $X = (\tilx) E$ if 
$P$ and $E \sub FX$ are ground bisimilar -- a single ground instance of the equation 
is evaluated.


\subsubsection*{The full \pc}

In the full $\pi$-calculus, including the output prefix $\out a
{\tilb}. P$,  ground bisimilarity is not a
congruence. One has therefore to use other forms of bisimilarity.
The most used is  \emph{early bisimilarity}; correspondingly one uses
an early transition system, where the parameters of inputs 
are instantiated with arbitrary (i.e., not necessarily
fresh) names, as in $\inp a{\til{x}}.P \arr{\iae a \tilb } P
\sub{\tilb}{\til x} $.
Modulo the move to the early setting, the theory exposed for the
asynchronous $\pi$-calculus also holds in the full $\pi$-calculus. 

However, when considering a different LTS, such as the early LTS,
divergences that did not exist in the ground LTS may arise, and as
such, prevent Theorems~\ref{thm:roscoe_api} or~\ref{thm:roscoe2_api}
to be used.  It is even possible that, by considering a different LTS,
an equation that had a unique solution loses this property. We leave
the study of this question for future work.

\subsection{An application: encoding of the {call-by-name} $\lambda$-calculus  }
\label{s:application_cbn}


To show an extended 
application of our techniques for the 
$\pi$-calculus, we revisit the proof of full abstraction for Milner's encoding of the
call-by-name (or lazy) $\lambda$-calculus  into \Apc 
  \cite{milnercbn} with respect to L{\'e}vy
Longo Trees (LTs),  precisely the   completeness part. 
We use $M,N$ to range over
the set $\Lambda$ of $\lambda$-terms, and              
$x,y,z$ to range over 
$\lambda$ 
variables. 
\iflong

We assume the standard concepts of free and bound variables and
substitutions, and
 identify $\alpha$-convertible terms. 
The terms in $ \Lambda$ are sometimes call \emph{open} to distinguish them from the subset
of  \emph{closed}  terms -- those without free variables. 
The rules defining the call-by-name strategy, defined on open terms,
are the following:
\begin{mathpar}
  (\lambda x . M) N\rightarrow M\{N/x\} \and
\inferrule{   M\rightarrow M' }{   M~N   \rightarrow M'~N} 
\end{mathpar}

 As usual $ \Rightarrow $ is
the reflexive and transitive closure of the single-step reduction $ \rightarrow $.  

We write $M \Up$ if $M$ diverges. 
If $M$  is an open \lterm, then  either $M$ diverges, 
or  $M\Rightarrow\lambda x. M'$ (for some $x$ and $M'$),  or
$M\Rightarrow x~M_1~\dots~M_n$ (for some $x$, $M_1$, \ldots, $M_n$).


The following notion of tree is used to provide a semantics for the
call-by-name $\lambda$-calculus.
\begin{definition}[L{\'e}vy-Longo Tree]
 The \emph{L{\'e}vy-Longo Tree} (LT) of an open $\lambda$-term $M$, written $LT(M)$,  is the 
(possibly infinite) tree defined coinductively as follows.
\begin{enumerate}
\item If $M$ diverges, then $LT(M)$ is the tree with a single node labelled $\bot$.
\item If $M\Rightarrow\lambda x.M'$,  
then $LT(M)$ is the tree with a root labelled with "$\lambda x.$", and 
$LT(M')$ as its
 unique descendant.
\item If $M\Rightarrow x~M_1~\dots~M_n$,
then $LT(M)$ is the tree
 with a root labelled with "$x$", and  $LT(M_1),\dots,LT(M_n)$ (in this order) as its  $n$ descendants.
\end{enumerate}
\end{definition}

LT equality 
\iflong
(whereby two $\lambda$-terms are identified if their LTs are equal) 
\else
(two $\lambda$-terms are identified if their LTs are equal) 
\fi
can
also be presented as a bisimilarity (\emph{open bisimilarity}, $\bsimo$),  defined as the largest
\emph{open bisimulation}.   
\begin{definition}
\label{d:open}
A relation $\R$ on $\Lambda $ is an \emph{open bisimulation} if,
whenever $M\RR N$:
\begin{enumerate}
\item\label{d:open:conv} $M\Rightarrow \lambda x. M'$ implies $N\Rightarrow \lambda
  x. N'$ with $M' \RR N'$;
\item\label{d:open:red} $M\Rightarrow x~M_1~\dots~M_n$ with $n\geq 0$ implies $N\Rightarrow x~N_1~\dots~N_n$
  and $M_i  \RR N_i$ for all $1\leq i\leq n$.
\item The converse of clauses \ref{d:open:conv} and \ref{d:open:red} on the challenges from
$N$.
\end{enumerate}
\end{definition}

\iflong
\begin{thmC}[\cite{cbn,cbn2}]\label{t:lt:openbis}
Let $M$ and $N$ be two \lterms; then $LT(M)=LT(N)$ iff $M\bsimo N$.
\end{thmC}
\fi

Milner's encoding of the call-by-name $\lambda$-calculus into \Apc~\cite{milnercbn}
is defined as follows: 
\begin{mathpar}
\enco {\lambda x. M}  ~\Defi~ (p) ~ p (x,q).~\app{\enco M} q
\and
    \enco x~\Defi~  (p)~\out x p 
\and
\enco
      {M~N} ~\Defi~  (p)~ \res {r, x} (\app{\enco{ M}} r | 
\out r {x,p}|!x(q).\app{\enco{   N }}q)
\end{mathpar}

Function application  
is translated into a particular form of 
parallel combination of two agents, the function and its argument;  
$\beta$-reduction  is then  modeled as   process interaction. 
Since the syntax of the  $\pi$-calculus  only allows for the
transmission of names 
along channels, the communication of a term is simulated by the
communication of a {\em trigger\/} for it. 
The translation of a $\lambda$-term is an abstraction that is parametric on a name, the
\emph{location} of the $\lambda$-term, which is intuitively the name along which the
term, as a function, will receive its argument. 
Precisely, the encoding of a term 
receives two names along its location       $p$: the first is 
 a
trigger  for its argument and  the second is  the  location to be
used for 
 the next interaction.

\smallskip
The full abstraction theorem for the encoding \cite{cbn,cbn2} states that 
two $\lambda$-terms have the same LT iff their encodings into \Apc
 are
weakly bisimilar terms. 
Full abstraction has two components: soundness, which says that
if the encodings are weakly bisimilar then the original terms  have the same LT; and
completeness, which is the converse direction. 
The 
\iflong
proof~\cite{cbn,cbn2} 
\else
proof~\cite{cbn} 
\fi
first establishes an operational correspondence between
the behaviour (visible and silent actions) of $\lambda$-terms and  of their encodings.
Then, exploiting this correspondence, soundness and completeness are
proved  
using the bisimulation proof method. For soundness, this amounts to
following the defining clauses of open
bisimulation 
(Definition~\ref{d:open}). In contrast, 
 completeness  involves enhancements of the bisimulation proof
method, notably  `bisimulation up to context and expansion'.  
\iflong
In the  latter, \emph{expansion}  is an auxiliary preorder relation, finer than weak
bisimilarity.
\fi
 As a consequence,  the technique requires having
developed the basic theory for the expansion preorder (e.g., precongruence properties and basic algebraic
laws), and requires
an operational correspondence  fine enough in order to be
able to reason about expansion
\iflong
 (expansion appears within the statements of operational
correspondence). 
\fi

Below we show that, by appealing to unique solution of equations,  completeness can be
proved by defining   an appropriate system of equations, each equation
having a simple shape,
and without the need for auxiliary preorders. 
For this, the only  results needed
are: $(i)$ validity of $\beta$-reduction for
the encoding (Lemma~\ref{l:beta}), whose
 proof is simple and consists in the application of a  few algebraic
 laws (including  laws for replication); $(ii)$ the property that
 if $M$ diverges then $\app{\enco M}p $ may never produce a visible
 action (Lemma~\ref{l:enco_div}); $(iii)$ a Lemma for rearranging parallel composition 
 and restrictions in the process encoding  
 $x~M_1~\dots~M_n$ (Lemma~\ref{l:enco_aux}).  

\begin{lemma}[Validity of $\beta$-reduction, \cite{cbn}]
\label{l:beta}
For $M \in \Lambda$, if $M \rightarrow M'$ then $\enco M \wb \enco{M'}$.   
\end{lemma} 

\begin{lemma}
  \label{l:enco_div}
If $M$ diverges, then $\enco M \bsim \abs p \nil$.
\end{lemma}

\begin{lemC}[\cite{cbn}]\label{l:enco_aux}
For any $M_1,\dots,M_n\in\Lambda$, and variable $x$, we have
\begin{align*}
\enco {x~M_1~\dots~M_n} & = \abs p 
( \res{ r_0, \ldots, r_n} )\big(\out x {r_0} | \out  {r_0}{r_1,x_1} | \ldots |
\out  {r_{n-1}}{r_n,x_n}  | \\
 & \quad\quad\quad\quad\quad\quad\quad
  ! \inp{x_1}{q_1} . \app{\enco {M_1}}{q_1} |  \ldots
| 
   ! \inp{x_n}{q_n} . \app{\enco{M_n}}{q_n}\big)
   \enspace.
\end{align*}
\end{lemC}

\begin{theorem}[Completeness, \cite{cbn}]\label{t:fullabs}
For $M, N\in\Lambda$, 
 $LT(M) = LT(N)$ implies $\enco M \bsim \enco N$.
\end{theorem}

\begin{proof}
Suppose $M_0$ and $N_0$ are two 
$\lambda$-terms with the same LT. 
We  define a system of equations, whose solutions  
are obtained from the encodings of $M_0$ and $N_0$.
We will then use 
Theorem~\ref{thm:roscoe2_api} to deduce 
$\enco {M_0} \wb \enco{N_0}$. 

Since $M_0$ and $N_0$ have the same LT, then by Theorem~\ref{t:lt:openbis} there is an open bisimulation $\R$ containing the
pair $(M_0,N_0)$. 
The variables of the equations are of the form $X_{M,N}$ for $M\RR
N$, and there is one equation for each pair in $\R$. 

 We  consider a pair $(M,N)$ in $\RR$, and explain how the
 corresponding equation is defined. 
We assume an ordering of the $\lambda$-calculus  variables so to be able to view
a finite set of variables as a tuple. This allows us to write  $\tilx$ for the   variables appearing free in $M$ or  $N$.

The equations are  the  translation of the clauses of
Definition~\ref{d:open}, assuming a generalisation of the encoding to
equation variables by adding the clause: 
\hskip .2cm $ 
\enco {X_{M,N}} \defi (\tilx, p)  \app {X_{M,N}}{\tilx, p}
$.

Since  $M \RR N$,  then either  both $M$ and $N$ diverge, or they satisfy
 one of the two clauses of Definition~\ref{d:open}.

 \begin{itemize}
 \item If  $M, N$  are both divergent, then the equation is 
$
X_{M,N}  = (\tilx,p) !\tau 
$. 
(as $!\tau\bsim \enco{\Omega}$)

\item If  $M, N$ satisfy clause~\ref{d:open:conv} 
of Definition~\ref{d:open}, the
equation  is 
\[ 
X_{M,N}  = \abs{\tilx,p}  \app{ \enco{\lambda x. X_{M',N'}}} p
\enspace,
\quad\mbox{ that is,  }\quad
X_{M,N}  = \abs{\tilx,p}  p(x,q ). \app{X_{M',N'}}{\til{y},q }  
\enspace,
\] 
where 
$\til {y}$ are the free variables in $M',N'$. 

\item 
If  $M, N$ satisfy clause~\ref{d:open:red} 
of Definition~\ref{d:open}, the equation is 
given by the translation of $x~ X_{M_1,N_1}\ldots X_{M_n,N_n}$, which,
rearranging restrictions and 
parallel compositions (Lemma~\ref{l:enco_aux}), can be written
\[
X_{M,N}  = (\tilx,p) 
( \res{ r_0, \ldots, r_n} )\big(\begin{array}[t]{l}
 \out x {r_0} | \out  {r_0}{r_1,x_1} | \ldots |
\out  {r_{n-1}}{r_n,x_n}  | \\
 ! \inp{x_1}{q_1} . \app{X_{M_1,N_1}}{\til {x_1},q_1} |  \ldots
| 
! \inp{x_n}{q_n} . \app{X_{M_n,N_n}}{\til {x_n},q_n}
                              \,\big)
\end{array}
 \]
where $\til {x_i}$ are the free variables in $M_i,N_i$. 
 \end{itemize}
 Essentially, each equation above represents the translation of a specific
 node of the LT for $M$ and $N$.

We then rely on Lemma~\ref{l:div_api} 
to show that 
 the equations may only produce innocuous divergences.
 It is easy to check that the syntactic condition holds:
  a   location name may only
appear once (in input position);  a trigger name either appears once (as a replicated
input), or it only appears in output position. 

Now, for $(M,N)\in\R$, 
define $F_{M,N} $ as the abstraction $(\tilx,p)\app{ \enco {M}}{p}$, and similarly 
 $G_{M,N} \Defi (\tilx,p)\app{ \enco {N}}{p}$.
Lemma~\ref{l:beta} allows us to show that the set of all such abstractions $F_{M,N} $ yields a solution for the
system of equations, and similarly for $G_{M,N}$.  

We reason by cases, following Definition~\ref{d:open}.
\begin{itemize}
\item 
  If clause (1) of
Definition~\ref{d:open} holds, then $M\Rightarrow x~M_1~\dots~M_n$ and we have 
\begin{align*}
F_{M,N} & =  (\til y,p) \app{\enco M}p \\
& \wb  (\til y,p) \app{\enco {x M_1\dots M_n}}p & \tag*{(by Lemma~\ref{l:beta})}\\
& = (\til y,p)p(x,q).\enco{M'}q\\
& =  (\til y,p) \big(p(x,q).\app {X_{M',N'}}{\til z}\big)\{F_{M',N'}/X_{M',N'}\}
\end{align*}
where $\til z$ is a subset of $x,\til y$, containing the free variables of $M$ and $N$.

\item If clause (2) holds, then $M\Rightarrow \lambda x. M'$ and we have
\begin{align*}
F_{M,N} & =  (\til y,p) \app{\enco M}p \\
& \wb  (\til y,p) 
( \res{ r_0, \ldots, r_n} )\big(\out x {r_0} | \out  {r_0}{r_1,x_1} | \ldots |
\out  {r_{n-1}}{r_n,x_n}  | \\
&\quad\quad ! \inp{x_1}{q_1} . \enco{M_1}{q_1} |  \ldots
| 
! \inp{x_n}{q_n} . \enco{M_n}{q_n}
                              \,\big) & \tag*{(by Lemmas~\ref{l:beta}
                                        and~\ref{l:enco_aux})}\\ 
&\bsim  (\til y,p) 
( \res{ r_0, \ldots, r_n} )\big(\out x {r_0} | \out  {r_0}{r_1,x_1} | \ldots |
\out  {r_{n-1}}{r_n,x_n}  | \\
&\quad\quad ! \inp{x_1}{q_1} . \app{X_{M_1,N_1}}{\til {x_1},q_1} |  \ldots
| 
! \inp{x_n}{q_n} . \app{X_{M_n,N_n}}{\til {x_n},q_n}
                              \,\big) \\
&\quad \quad \{F_{M_1,N_1}/X_{M_1,N_1},\dots,F_{M_n,N_n}/X_{M_n,N_n}\}
\end{align*}
\item 
  If neither clause (1) or (2) hold, meaning both $M$ and $N$ diverge, 
 then we have 
\begin{align*}
F_{M,N} &= (\tilx,p) \app {\enco M}p\\
&\bsim (\tilx,p) \nil&\tag*{(by Lemma~\ref{l:enco_div})}\\
&\bsim (\tilx,p) !\tau
\end{align*}

\end{itemize} 

Since $F_{M,N}$ and $G_{M,N}$ are
both solutions of the system of equations, they are bisimilar, which
finally allows us to deduce that $\enco {M_0} \wb \enco{N_0}$
\end{proof}

\section{Conclusions \& Future Work}

We have compared our techniques to one of the most powerful forms of  enhancements of the
bisimulation proof method, namely Pous `up to transitivity
and context', showing that, up to a technical condition, our techniques are at least as
powerful.  
 We believe that also the converse holds, though possibly
under different side conditions. We leave a detailed analysis of this comparison, which seems
non-trivial,  for
future work. 
In this respect, the goal of 
the work on unique solution of equations is to provide a
way of better understanding   up-to techniques
and to shed  light into the conditions for their
soundness.  The technique by Pous, in particular, 
is arguably more complex, 
both in its definition and its application,
\iflong
 than the 
unique solution theorems presented here. 
\fi

As said above, the comparison with up-to techniques could also help
understanding the need for 
the
 closure under substitutions in up to context techniques for name-passing calculi such as 
the asynchronous
$\pi$-calculus.     
\iflong
Surprisingly, our unique solution techniques, despite the strong similarities 
with up-to context
techniques, do not require the closure under substitutions. 
\fi
However, currently it is unclear how to formally relate 
 bisimulation enhancements and `unique solution of
equations' in name-passing calculi.
\ifnot
\finish{or below?}
whether the comparison
might shed light into the open problem for the $\pi$-calculus, as the
relationship between bisimulation enhancements and unique solution of
equation theorems appears to be weaker in name-passing calculi with
respect to CCS-like  languages.
\fi

Up-to techniques have been analysed in an abstract
setting using lattice theory~\cite{DBLP:conf/lics/Pous16} and category
theory~\cite{DBLP:journals/acta/BonchiPPR17,DBLP:conf/sofsem/RotBR13}.
It could  be interesting to do the same for the unique-solution
techniques, 
to study their connections with up-to techniques, and
to understand which equivalences can be handled 
(possibly using, or refining, the   abstract formulation 
\iflong presented in \else of \fi 
Section~\ref{s:generalisation}).

In comparison with the enhancements of the bisimulation proof method,
the main drawback of the techniques exposed in this paper 
 is the presence of a semantic condition,
involving divergence: the unfoldings of the
equations should not produce divergences, or only produce innocuous divergences.
A syntactic condition for this has been proposed (Lemma~\ref{l:criterion:div}). 
Various techniques for checking
divergence in concurrent calculi exist in the literature, including 
type-based techniques
\cite{yoshida:berger:honda:termination:ic,DBLP:conf/concur/DemangeonHS10,DBLP:journals/iandc/DengS06}.
 However, in general
divergence is undecidable, and therefore, the check may sometimes be
unfeasible.  Nevertheless, the equations that one writes for proofs
usually involve forms of `normalised' processes, and as such they are
divergence-free (or at most, contain only innocuous divergences).
More experiments are needed to validate this claim
or to understand how limiting this problem is.



Several studies in functional
programming and type theory 
rely on type-based 
methods to insure that coinductive definitions are productive, i.e.,
do not give rise to partially defined functions (in order to preserve
logical consistency)~\cite{DBLP:conf/lics/Nakano00,DBLP:conf/icfp/AtkeyM13}. Understanding whether these approaches can be
adapted to analyse divergences and innocuous divergences in systems of
equations is a topic for future investigations.

\subparagraph*{Acknowledgements.}

This work was supported by Labex MILYON/ANR-10-LABX-0070, by
 the
European
Research Council (ERC) under the Horizon 2020 programme (CoVeCe,
grant agreement No {678157}),
 by  H2020-MSCA-RISE project `Behapi' (ID 778233), and by the
project ANR-16-CE25-0011 REPAS.


\bibliographystyle{alpha}
\bibliography{main}



\end{document}